\documentclass{article}

\usepackage[margin=0.9in]{geometry}

\usepackage{xspace}
\usepackage{verbatim}
\usepackage{color}
\usepackage{graphicx}
\usepackage{tabularx}
\usepackage{booktabs}
\usepackage{amsmath}
\usepackage[showonlyrefs]{mathtools}
\usepackage{empheq}
\usepackage{amstext,amssymb,amsfonts}

\usepackage{nicefrac}
\usepackage{bm}
\usepackage{bbm}

\usepackage[shortlabels]{enumitem}
\usepackage[ruled,vlined,linesnumbered]{algorithm2e}

\usepackage{todonotes}

\newcounter{mynotes}
\newcommand{\mnote}[1]{\addtocounter{mynotes}{1}{{}}%
\todo[color=blue!20!white]{[\arabic{mynotes}] \scriptsize  {\small {\sf {#1}}}}}

\usepackage{textcomp,setspace}

\usepackage{nameref}
\usepackage[pagebackref,colorlinks,linkcolor=blue,filecolor = blue,citecolor = blue, urlcolor = blue, hyperfootnotes=false]{hyperref}
\usepackage{color}

\usepackage{cleveref} 
\renewcommand{\cref}{\Cref}

\usepackage{amsthm}
\usepackage{thmtools}
\usepackage{thm-restate}

\declaretheorem[within=section]{theorem}
\declaretheorem[sibling=theorem]{corollary}

\declaretheorem[sibling=theorem]{lemma}

\declaretheorem[sibling=theorem]{definition}

\declaretheorem[sibling=theorem]{remark}

\declaretheorem[sibling=theorem]{proposition}

\newcounter{termcounter}
\renewcommand{\thetermcounter}{\Alph{termcounter}}
\crefname{term}{term}{terms}
\creflabelformat{term}{\textup{(#2#1#3)}}

\makeatletter
\def\term{\@ifnextchar[\term@optarg\term@noarg}
\def\term@optarg[#1]#2{%
  \textup{(#1)}%
  \def\@currentlabel{#1}%
  \def\cref@currentlabel{[][2147483647][]#1}%
  \cref@label[term]{#2}}
\def\term@noarg#1{%
  \refstepcounter{termcounter}%
  \textup{(\thetermcounter)}%
  \cref@label[term]{#1}}
\makeatother

\newcommand{\ignore}[1]{}

\newcommand{\poly}{\mathrm{poly}}

\renewcommand{\vec}[1]{{\bf{#1}}}

\definecolor{DSred}{rgb}{1,0,0}

\renewcommand{\leq}{\leqslant}
\renewcommand{\geq}{\geqslant}
\renewcommand{\ge}{\geqslant}
\renewcommand{\le}{\leqslant}
\renewcommand{\epsilon}{\varepsilon}
\newcommand{\eps}{\epsilon}

\newcommand{\R}{\mathbb{R}}

\newcommand{\Z}{\mathbb{Z}}

\newcommand{\F}{\mathbb{F}}

\newcommand{\Esymb}{{\bf E}}

\newcommand{\Psymb}{{\bf Pr}}

\DeclareMathOperator*{\E}{\Esymb}

\DeclareMathOperator*{\ProbOp}{\Psymb}

\renewcommand{\Pr}{\ProbOp}

\newcommand{\ComplexityFont}[1]{\ensuremath{\mathsf{#1}}}

\newcommand{\NP}{\ComplexityFont{NP}}

\usepackage{mdwlist}
\allowdisplaybreaks

\theoremstyle{definition}

\DeclareMathOperator*{\argmin}{\arg\,min}
\newcommand{\Sp}{\mathsf{Sp}}
\title{Testing Sparsity over Known and Unknown Bases}
\author{Siddharth Barman\thanks{Department of Computer Science and Automation, Indian Institute of Science, Bangalore, India. Emails: {\tt \{barman, arnabb, suprovat\}@iisc.ac.in}.}
	\and Arnab Bhattacharyya$^*$
	\and Suprovat Ghoshal$^*$ 
}
\date{}	

\begin{document}

\maketitle
	
	\begin{abstract}
		
 Sparsity is a basic property of real vectors that is exploited in a wide variety of applications. In this work, we describe property testing algorithms for sparsity that observe a low-dimensional projection of the input. 

We consider two settings. In the first setting, for a given design matrix $\vec{A} \in \R^{d \times m}$, we test whether an input vector ${\bf y} \in \R^d$ equals ${\bf Ax}$ for some $k$-sparse unit vector ${\bf x}$. Our algorithm projects the input onto $O(k\eps^{-2}\log m)$ dimensions, accepts if the property holds, rejects if $\|{\bf y} - {\bf Ax}\|> \eps$ for any $O(k/\eps^2)$-sparse vector ${\bf x}$, and runs in time nearly polynomial in $m$. Our algorithm is based on the approximate Carath\'eodory's Theorem. Previously known algorithms that solve the problem for arbitrary $\vec{A}$ with qualitatively similar guarantees run in exponential time.

In the second setting, the design matrix ${\bf A}$ is unknown. Given input vectors ${\bf y}_1, \dots, {\bf y}_p \in \R^d$ whose concatenation as columns forms ${\bf Y} \in \R^{d \times p}$, the goal is to decide whether ${\bf Y} = {\bf A}{\bf X}$ for matrices ${\bf A} \in \R^{d \times m}$ and ${\bf X} \in \R^{m \times p}$ such that each column of ${\bf X}$ is $k$-sparse, or whether ${\bf Y}$ is ``far'' from having such a decomposition. We give such a testing algorithm which projects the input vectors to $O((\log p)/\eps^2)$ dimensions and assumes that the unknown ${\bf A}$ satisfies {\em $k$-restricted isometry}. Our analysis gives a new robust characterization of {\em gaussian width} in terms of sparsity.

	\end{abstract}

\section{Introduction}
{\em Property testing} is the study of algorithms that query their input a small number of times and distinguish between whether their input satisfies a given property or is ``far" from satisfying that property. The quest for efficient testing algorithms was initiated by \cite{BLR} and \cite{BFL} and later explicitly formulated by \cite{RS} and \cite{GGR}. Property testing can be viewed as a relaxation of the traditional notion of a decision problem, where the relaxation is quantified in terms of a distance parameter.
There has been extensive work in this area  over the last couple of decades; see, for instance, the surveys \cite{RonSurvey08} and \cite{RubinfeldICM} for some different perspectives. 

As evident from these surveys, research in property testing has largely focused on properties of combinatorial and algebraic structures, such as bipartiteness of graphs, linearity of Boolean functions on the hypercube, membership in error-correcting codes or representability of functions as concise Boolean formulae. In this work, we study the question of testing properties of {\em continuous} structures, specifically properties of vectors and matrices over the reals.

Our computational model extends the standard property testing framework by allowing queries to be linear measurements of the input. Let $\mathcal{P} \subset \R^d$ be a property of real vectors. Let $\mathsf{dist}:\R^d \to \R^{\geq 0}$ be a ``distance" function such that $\mathsf{dist}({\bf x})=0$ for all ${\bf x} \in \mathcal{P}$. We say that an algorithm $\mathcal{A}$ is a {\em tester for $\mathcal{P}$} with respect to $\mathsf{dist}$ and with parameters $\eps, \delta>0$ if for any input ${\bf y} \in \R^n$, the algorithm $\mathcal{A}$ observes ${\bf My}$ where $\vec{M} \in \R^{q \times d}$ is a randomized matrix and has the following guarantee:
\begin{enumerate}
\item[(i)]
If ${\bf y} \in \mathcal{P}$, $\Pr_{\vec M}[\mathcal{A}({\bf My}) \text{ accepts}]\geq 1-\delta$.
\item[(ii)]
If $\mathsf{dist}({\bf y}) > \eps$, $\Pr_{\vec M}[\mathcal{A}({\bf My}) \text { accepts}]\leq \delta$.
\end{enumerate}
We call each inner product between the rows of $\vec{M}$ and $\vec{y}$ a {\em (linear) query}, and the number of rows $q=q(\eps,\delta)$ is the {\em query complexity} of the tester. The {\em running time} of the tester $\mathcal{A}$ is its running time  on the outcome of its queries. As typical in property testing, we do not count the time needed to evaluate the queries.
If $\mathcal{P} \subset \R^{d \times p}$ is a property of real matrices with an associated distance function $\mathsf{dist}: \R^{d \times p} \to \R^{\geq 0}$, testing is defined similarly: given an input matrix $\vec{Y} \in \R^{d \times p}$, the algorithm observes $\vec{MY}$ for a random matrix $\vec{M} \in \R^{q \times d}$ with analogous completeness and soundness properties. 
A linear projection of an input vector or matrix to a low-dimensional space is also called a {\em linear sketch} or a {\em linear measurement}. The technique of obtaining small linear sketches of high-dimensional vectors has been used to great effect in algorithms for streaming (e.g., \cite{AMS98, McGregor}) and numerical linear algebra (see \cite{Woo} for an excellent survey). 

We focus on testing whether a vector is {\bf sparse} with respect to some basis.\footnote{With slight abuse of notation, we use the term basis to denote the set of columns of a design matrix. The columns might not be linearly independent.} A vector ${\bf x}$ is said to be {\em $k$-sparse} if it has at most $k$ nonzero coordinates. Sparsity is a structural characteristic of signals of interest in a diverse range of applications. It is a pervasive concept throughout modern statistics and machine learning, and algorithms to solve inverse problems under sparsity constraints are among the most successful stories of the optimization community (see the book \cite{HTW}). The natural property testing question we consider is whether there exists a solution to a linear inverse problem under a sparsity constraint.  

There are two settings in which we investigate the sparsity testing problem. 
\begin{enumerate}
\item[(a)]
In the first setting, a design matrix $\vec{A} \in \R^{d \times m}$ is known explicitly, and the property to test is whether a given input vector ${\bf y} \in \R^d$ equals ${\bf Ax}$ for a $k$-sparse unit vector ${\bf x} \in \R^m$. For instance, $\vec{A}$ can be the Fourier basis or an overcomplete dictionary in an image processing application. We approach this problem in full generality, without putting any restriction on the structure of $\vec{A}$.

Informally, our main result in this setting is that for any design matrix $\vec{A}$, there exists a tester projecting the input ${\bf y}$ to $O(k \log m)$ dimensions that rejects if ${\bf y}-{\bf Ax}$ has large norm for any $O(k)$-sparse ${\bf x}$. The running time of the tester is polynomial in $m$. As we describe in Section \ref{sec:relwork}, previous work in numerical linear algebra yields a tester with the same query complexity and with qualitatively similar soundness guarantees but which requires running time {\em exponential} in $m$.

\item[(b)]
In the second setting, the design matrix $\vec{A}$ is not known in advance. For input vectors ${\bf y}_1, {\bf y}_2, \dots, {\bf y}_p \in \R^d$, the property to test is whether there exists a matrix $\vec{A} \in \R^{d \times m}$ and $k$-sparse unit vectors ${\bf x}_1, {\bf x}_2, \dots {\bf x}_p \in \R^m$ such that ${\bf y}_i = {\bf Ax}_i$ for all $i \in [p]$. Note that $m$ is specified as a parameter and could be much larger than $d$ (the {\em overcomplete} case). In this setting, we restrict the unknown $\vec{A}$ to be a {\em $(\eps, k)$-RIP} matrix which means that $(1-\eps)\|{\bf x}\| \leq \|{\bf Ax}\| \leq (1+\eps) \|{\bf x}\|$ for any $k$-sparse ${\bf x}$. This is a standard assumption made in many related works (see Section \ref{sec:relwork} for details). 

In this setting, we design an efficient tester for this property that projects the inputs to $O(\eps^{-2} \log p)$ dimensions and, informally speaking, rejects if for all $(\eps, k)$-RIP matrices $\vec{A}$, there is some ${\bf y}_i$ such that ${\bf y}_i-{\bf Ax}_i$ has large norm for all ``approximately sparse'' ${\bf x}_i$.
\end{enumerate}

In both of the above tests, the measurement matrix is a random matrix with iid gaussian entries, chosen so as to preserve norms and certain other geometric properties upon dimensionality reduction.\footnote{If evaluating the queries efficiently was an objective, one could also use sparse dimension reduction matrices \cite{DasguptaKS, KaneN,BourgainN}, but we do not pursue this direction here.} In particular, our testers are {\em oblivious} to the input. It is a very interesting open question as to whether non-oblivious testers can strengthen the above results. 

\subsection{Our Results}

%
%
%

We now present our results more formally. For integer $m>0$, let $\mathcal{S}^{m-1} = \{\vec{x} \in \R^m: \|\vec{x}\|=1\}$, and let $\Sp_k^m = \{{\bf x} \in \mathcal{S}^{m-1} : \|{\bf x}\|_0 \leq k\}$.\footnote{Here, $\| {\bf x}\|_0 $ denotes the the sparsity of the vector, $\| {\bf x}\|_0 := |\{i \in [m] \mid x_i \neq 0 \}|$.}
\ignore{
Our main results are construction and analysis of efficient algorithms for the testing sparsity for the known and unknown basis settings. We formally state the result for the known basis setting, followed by the unknown basis setting.}

\begin{restatable}[Known Design Matrix]{thm}{known}  \label{thm:test_known}
	Fix $\eps,\delta \in (0,1)$ and positive integers $d, k, m$ and a matrix ${\bf A} \in \mathbbm{R}^{d \times m}$ such that $\|{\bf a}_i\| = 1$ for every $i \in [m]$. There exists a tester with query complexity $O(k\eps^{-2}\log(m/\delta))$ that behaves as follows for an input vector ${\bf y} \in \R^d$:
	\begin{itemize}
		\item \emph{\textbf{Completeness}}: If ${\bf y} = {\bf A}{\bf x}$ for some ${\bf x} \in \Sp_k^m$, then the tester accepts with probability $1$.
		
		\item \emph{\textbf{Soundness}}: If  $\|{\bf A}{\bf x} - {\bf y}\|_2 > \eps$ for every $\vec{x} : \|\vec{x}\|_0 \le K$, then the tester rejects with probability $\ge 1 - \delta$. Here, $K = O(k/\eps^2)$.
	\end{itemize}
	The running time of the tester is $\poly(m, k, 1/\eps)$.
\end{restatable}

The tester for the known design case approximates $\vec{y}$ as a sparse convex combination of the vertices of a low-dimensional polytope. This connection between the approximate Carath\'eodory problem and sparsity-constrained linear regression may be useful in other contexts too.

\ignore{On the other hand, it is known that $\Omega(k\log m)$-linear measurements are necessary for the problem when the underlying sparse vector ${\bf x} \in \mathbbm{R}^m$ has to be recovered \cite{BIPW10}. Furthermore, the only assumption we make on the matrix ${\bf A}$ is that the columns should be of unit norm. This is in sharp contrast to known efficient algorithms for sparse recovery, where the measurement matrix ${\bf A}$ is assumed to satisfy the conditions such as RIP or incoherence.}

We now describe our result for the unknown design matrix.

\begin{restatable}[Unknown Design Matrix]{thm}{unknown}  \label{thm:test_rip}
	Fix $\eps, \delta \in (0,1)$ and positive integers $d, k, m$ and $p$, such that $(k/m)^{1/8} < \eps < \frac{1}{100}$ and $ k \geq 10\log\frac{1}{\epsilon}$.
	There exists a tester with query complexity $O(\epsilon^{-2}\log{(p/\delta)})$ which, given as input
	vectors $\vec{y}_1, \vec{y}_2, \dots, \vec{y}_p \in \R^d$,  has the following behavior (where ${\bf Y}$ is the matrix having $\vec{y}_1, \vec{y}_2, \dots, \vec{y}_p$ as columns):
	\begin{itemize}
		\item \emph{\textbf{Completeness}}: If ${\bf Y}$ admits a decomposition $ {\bf Y} = {\bf A}{\bf X}$, where ${\bf A} \in \R^{d \times m}$  satisfies $(\eps, k)$-RIP  and ${\bf X} \in \R^{m \times p}$ with each column of ${\bf X}$ in $\Sp_k^m$, then the tester accepts with probability $\ge 1 - \delta$.
		
		\item \emph{\textbf{Soundness}}: Suppose ${\bf Y}$ does not admit a decomposition ${\bf Y} = {\bf A}({\bf X} + {\bf Z}) + {\bf W}$ with 
		\begin{itemize}
			\item[1.] The design matrix ${\bf A} \in \mathbbm{R}^{d \times m}$ being $(\epsilon,k)$-RIP, with $\|{\bf a}_i\| = 1$ for every $i \in [m]$.
			\item[2.] The coefficient matrix ${\bf X} \in \mathbbm{R}^{m \times p}$ being column wise $\ell$-sparse, where $\ell = O(k/\eps^4)$.
			\item[3.] The error matrices ${\bf Z} \in \mathbbm{R}^{m \times p}$ and ${\bf W} \in \mathbbm{R}^{d \times p}$ satisfying
			\begin{equation*}
				\|{\bf z}_i\|_\infty \le \epsilon^2, \qquad \|{\bf w}_i\|_2 \le O(\epsilon^{1/4}) \qquad \text{for all } i \in [p]. 
			\end{equation*}
		\end{itemize}	
		Then the tester rejects with probability $\ge 1 - \delta$. 
	\end{itemize}
\end{restatable}

The contrapositive of the soundness guarantee from the above theorem states that if the tester accepts, then matrix ${\bf Y}$ admits a factorization of the form ${\bf Y} = {\bf A}({\bf X}+{\bf Z}) + {\bf W}$, with error matrices ${\bf Z}$ and ${\bf W}$ having $\ell_\infty$ and $\ell_2$ error bounds. The matrix ${\bf X} + {\bf Z}$ is a sparse matrix with $\ell_\infty$-based soft thresholding, and ${\bf W}$ is an additive $\ell_2$-error term.\footnote{Theorem~\ref{thm:test_rip} can be restated in terms of \emph{incoherent} (instead of RIP) design matrices as well. This follows from the fact that the incoherence and RIP constants of a matrix are order-wise equivalent. This observation is formalized in Appendix~\ref{appendix:incoherence-rip}.}

\begin{remark}[Problem Formulation]
	Note that the settings considered in the known and unknown design matrix settings are quite different from each other. In particular, for the known design setting, the input is a single vector. However, given a single input vector ${\bf y} \in \mathbbm{R}^d$, the analogous unknown design testing question for this setting would be moot, since one can always consider the vector ${\bf y}$ to be the design matrix ${\bf A}$, in which it trivially admits a $1$-sparse representation. More generally, this question is interesting only when the number of points $p$ exceeds $m$, by the same argument.
\end{remark}

\begin{remark}[Range of sparsity parameter $k$]
	It is important to note that the above problem is of interest only when $k < d$. This is true because any $S \subset \mathcal{S}^{d-1}$ trivially admits a $d$-sparse representation in any basis for $\mathbbm{R}^d$. \ignore{Furthermore using Lemma \ref{lem:appx-rank-bound}, we can actually guarantee the decomposition $Y = AX + W$ where each column of $W$ is entry-wise bounded by $\epsilon$ and $X$ is $O_\epsilon\Big(\log d ,\log p\Big)$-sparse, which again can be much smaller that $d$.\mnote{Changed $\omega^2(Y)$ to $\log p$ because gaussian width not defined yet. Also, this $\ell_\infty$ error inconsistent with above theorem, so not sure if this makes sense. Will come back to it.}} Therefore, the challenge here is to design a tester which works in the regime where $k$ is small.
\end{remark}  
	
The above tests have perfect completeness. In the property testing literature, testers with imperfect completeness are called {\em tolerant} \cite{ParnasRR}.
	We also give {tolerant} variants of these testers (Theorems \ref{thm:noise_known} and \ref{thm:test_rip_noisy8ye}) which can handle bounded noise for the completeness case. Finally, we also give an algorithm for testing dimensionality, which is based on similar techniques.

\begin{theorem}[Testing dimensionality]\label{thm:test_dim}
Fix $\eps, \delta \in (0,1)$, positive integers $d, k $ and $p$, where $k \ge 10\epsilon^2\log d$.
	There exists a tester with query complexity $O(\log \delta^{-1})$, which gives as input vectors $\vec{y}_1, \dots, \vec{y}_p \subset \mathcal{S}^{d-1}$, has the following behavior:
	\begin{itemize}
		\item
		\emph{\textbf{Completeness}}: If ${\rm rank}(Y) \leq k$, then the tester accepts with probability $\geq 1-\delta$.
		\item
		\emph{\textbf{Soundness}}: If ${\rm rank}_\eps(Y) \ge k'$, then the tester rejects with probability $\ge 1-\delta$. Here, $k' = 20k/\eps^2$
	\end{itemize}
\end{theorem}

The soundness criteria in the above Theorem is stated in terms of the $\epsilon$-approximate rank of a matrix (see Definition \ref{defn:appx_rank}). This is a well-studied relaxation of the algebraic definition of rank, and has applications in approximation algorithms, communication complexity and learning theory (see \cite{ALSV13} and references therein).


%

\subsection{Related Work}
\label{sec:relwork}

Although, to the best of our knowledge, the testing problems we consider have not been explicitly investigated before, there are several related areas of study that frame our results in their proper context.

\paragraph{Sketching in the Streaming Model.}
In the streaming model, one has a series of updates $(i,v)$ where each $i \in [n]$ and $v \in \{-T, \dots, T\}$. Each update modifies a vector $\vec{x}$, initialized at $\vec{0}$, to $\vec{x}+v\vec{e}_i$. The $L_0$-estimation problem in streaming is to estimate the sparsity of $\vec{x}$ upto a multiplicative $(1\pm \eps)$ factor. A linear sketch algorithm maintains $\vec{Mx}$ during the stream, where $\vec{M} \in \R^{s \times n}$ is a randomized matrix. 

A linear sketch algorithm for the $L_0$-estimation problem directly yields a tester in the setting where the design matrix is known to be the identity matrix. By invoking the space-optimal $L_0$-estimation result from \cite{KNW}, we obtain:
\begin{theorem}[Implicit in \cite{KNW}]
Fix $\eps \in (0,1)$, positive integers $m, k$ and an {\em invertible} matrix $\vec{A} \in \R^{m \times m}$. Then, there is a tester with query complexity $O(\eps^{-2} \log(m))$ that, for an input $\vec{y} \in \R^m$, accepts with probability at least $2/3$ if $\vec{y} = \vec{Ax}$ for some $k$-sparse $\vec{x} \in \Z^m$, and rejects with probability $2/3$ if 
$\vec{y}=\vec{Ax}$ for some $(1+\eps)k$-sparse $\vec{x} \in \Z^m$. The running time of the algorithm is $\poly(m, 1/\eps)$. 
\end{theorem}
We believe that the theorem should also extend (albeit with a mild change in parameters) to the setting where $\vec{x}$ is an arbitrary real vector (not necessarily discrete), but the assumption that $\vec{A}$ is invertible seems hard to circumvent.

\paragraph{Sketching in Numerical Linear Algebra.}
Low-dimensional sketches used in numerical linear algebra can also yield testers in the known design matrix case of our model.  For  a matrix $\vec{A} \in \R^{d \times m}$, suppose we want a property tester that, for input $\vec{y} \in \R^d$, distinguishes between the case $\vec{y} = \vec{Ax^*}$ for some $k$-sparse $\vec{x}^*$, and the case $\min_{k\text{-sparse } \vec{x}} \|\vec{Ax-y}\|>\eps$. 

In the sketching approach, one looks to solve the optimization problem $\min_{k\text{-sparse } \vec{x}} \|\vec{Ax-y}\|$ in a smaller dimension i.e., one looks at:
\begin{equation}
\widehat{\vec{x}} = \argmin_{\vec{x}' \in K} \|\vec{SAx}' - \vec{Sy}\| = \argmin_{\vec{x}' \in K} \|\vec{S}(\vec{Ax}' - \vec{y})\|
\label{eqn:hatmin}
\end{equation}
where $\vec{S} \in \R^{q \times d}$ is a sketch matrix (where $q \ll d$) and $K = \{\vec{x} : \|\vec{x}\|_0 \leq k\}$. The intent here is that the vector $\hat{\vec{\bf x}}$ would also be an approximate minimizer to the original optimization problem.   

An {\em oblivious $\ell_2$-subspace embedding} with parameters $(d, m, \eps, \delta)$ is a distribution on $q\times d$ matrices $\vec{M}$ such that with probability at least $1-\delta$, for any fixed $d \times m$ matrix $\vec{A}$, $(1-\eps)\|\vec{Ax}\| \leq \|\vec{MAx}\| \leq (1+\eps)\|\vec{Ax}\|$ for all $\vec{x} \in \R^m$. For our application, suppose we draw $\vec{S}$ from an oblivious subspace embedding with parameters{\footnote{That is, consider all possible choices of supports $\Omega \in {[m]\choose \le k}$ and let ${\bf A}_\Omega$ be the submatrix corresponding to the columns of $\Omega$. For the given choice of parameters, it follows that with probability $\ge 1 - \delta$, every ${\bf x} \in \R^m$ will satisfy $\|{\bf S}({\bf A}_{\Omega}{\bf x} - {\bf y})\| \in (1 \pm \epsilon)\|{\bf A}_{\Omega}{\bf x} - {\bf y}\|$}} $\big(d, k+1, \eps, \delta/{m \choose k}\big)$. Then, we get a valid property tester with query complexity $q$ if we accept when $\|\vec{SA\widehat{x} - Sy}\| = 0$ and reject when it is at least $\eps(1-\eps)$.

Using the oblivious subspace embedding from Theorem 2.3 in \cite{Woo}, we get the following theorem:

\ignore{Suppose $\vec{S}$ is drawn from an oblivious subspace embedding with parameters $d, m+1, \eps, \delta$, so that for any $\vec{x}$, $\|\vec{SAx - Sy}\| \in (1 \pm \eps) \|\vec{Ax-y}\|$ with probability at least $1-\delta$. Then, we get a valid property tester with query complexity $q$ if we accept when $\|\vec{SA\widehat{x} - Sy}\| = 0$ and reject when it is at least $\eps(1-\eps)$.

In our work for the known design matrix setting, we are looking at $K$ equal to the $\ell_0$-ball of radius $k$. In this case, it suffices to draw $\vec{S}$ from an oblivious subspace embedding with parameters $d, k+1, \eps, \delta/{m \choose k}$. Using the construction from Theorem 2.3 in \cite{Woo}, we get the following theorem:}
\begin{theorem}[Implicit in prior work]
Fix $\eps, \delta \in (0,1)$ and positive integers $d, k, m$ and a matrix $\vec{A} \in \R^{d \times m}$. Then, there is a tester with query complexity $O(k \eps^{-2} \log(m/\delta))$ that, for an input vector $\vec{y} \in \R^d$, accepts with probability $1$ if $\vec{y} = \vec{Ax}$ for some $k$-sparse $\vec{x}$ and rejects with probability at least $1-\delta$ if $\|\vec{y-Ax}\| > \eps$ for all $k$-sparse $\vec{x}$. The running time of the tester is the time required to solve Equation (\ref{eqn:hatmin}).
\end{theorem}
Unfortunately, for general design matrices $\vec{A}$, solving the optimization problem in Equation (\ref{eqn:hatmin}) is $\NP$-hard. When $\vec{A}$ satisfies $(\eps, k)$-RIP, then it is easy to verify that with probability at least $1-\delta$, $\vec{SA}$ is also $(O(\eps), k)$-RIP. In such cases, which as we explain next, Equation (\ref{eqn:hatmin}) can be solved efficiently, which in turn implies that the above property tester has polynomial running time.

\paragraph{Sparse Recovery and Compressive Sensing.}
{\em Sparse recovery} or {\em compressed sensing} is the problem of recovering a sparse vector $\vec{x}$ from a low-dimensional projection $\vec{Ax}$. In compressive sensing, $\vec{A}$ is interpreted as a measurement matrix, where each row of $\vec{A}$ corresponds to a linear measurement. Compressive sensing has been used for single-pixel cameras, MRI compression, and radar communication. See \cite{FoucartR} and references therein.

In celebrated works by Cand\`es, Romberg, and Tao \cite{CRT} and by Donoho \cite{Don}, it was shown that  given a matrix $\vec{A}$ satisfying $(0.4, k)$-RIP, any $k$-sparse vector $\vec{x}$ can be recovered efficiently from $\vec{y} =\vec{A} {\bf x}$, even when the sparsity $k = \Omega(d)$, and similar results hold when $\|\vec{y}-\vec{Ax}\|$ is small.  However, these results are not directly relevant to us as the recovery algorithms examine all of ${\bf y}$ and not just a low-dimensional sketch of it.

\begin{remark}
	Note that the sketching based approaches discussed in this subsection so far address the setting where the design matrix ${\bf A}$ is known, and as such do not have implications for the testing problem in the unknown design setting.
\end{remark}

\paragraph{Dictionary Learning.}
In the setting of the unknown design matrix, the question of recovering the design matrix and the sparse representation (as opposed to our problem of testing their existence) is called the {\em dictionary learning} or {\em sparse coding} problem. Dictionary learning is a fundamental task in several domains. The problem was first formulated by \cite{OF96, OF97} who showed that the dictionary elements learnt from sparse coding of natural images are similar to the receptive fields of neurons in the visual cortex. Inspired by these results, automatically learned dictionaries have been used in machine learning for feature selection by \cite{Evgen07} and for denoising by \cite{EA06}, edge-detection by \cite{MLBHP08}, super-resolution by \cite{YWHM08}, restoration by \cite{MSE08}, and texture synthesis by \cite{Pey09} in image processing applications.

\ignore{
Although the dictionary learning problem is \textsf{NP}-hard (\cite{DMA97}) in the worst case (even when the dictionary $A$ is known), it is generally considered a solved problem in practice. Gradient descent heuristics (\cite{OF97}), the method of optimal directions (\cite{EAHH99}) and the K-SVD algorithm (\cite{AEB05}) work very well for real applications; see \cite{Aha06} for a very nice overview. On the other hand, theoretical results have not yet been able to fully justify the efficacy of these methods. }
The first work to give a dictionary learning algorithm with provable guarantees was \cite{SWW12} who restricted the dictionary to be square and the sparsity to be at most $\sqrt{d}$. For the more common overcomplete setting, \cite{AGM14} and  \cite{AgarwalAltMin14} independently gave algorithms with provable guarantees for dictionaries satisfying incoherence and RIP respectively. These works also restrict the sparsity to be strictly less than $\sqrt{d}/\mu$ where $\mu$ is the incoherence. \cite{BKS15} gave a very different analysis using the sum-of-squares hierarchy that works for nearly linear sparsity; however, their algorithm runs in time $d^{\poly(1/\eps)}$ where $\eps$ measures the accuracy to which the dictionary is to be learned and this is too inefficient to be of use for realistic parameter ranges. All of these (as well as other more recent) works assume distributions from which the input samples are generated in an i.i.d fashion.
In contrast, our work is in the \emph{agnostic setting} and hence, is incomparable with these results.

\paragraph{Property Testing.}
We are not aware of any directly related work in the property testing literature. \cite{CSZ00} studied some problems in computational geometry from the property testing perspective, but the problems involved only discrete structures. Krauthgamer and Sasson \cite{KS03} studied the problem of testing dimensionality, but their notion of farness from being low-dimensional is quite different from ours. In their setup, a sequence of vectors $\vec{y}_1, \dots, \vec{y}_p$ is $\eps$-far from being $d$-dimensional if at least $\eps p$ vectors need to be removed to make it be of dimension $d$. Note that a set of vectors can be nearly isometric to a $d$-dimensional subspace but far from being $d$-dimensional in Krauthgamer and Sasson's sense (for example, the Johnson-Lindenstrauss projection of the standard unit vectors $\vec{e}_1, \vec{e}_2, \dots, \vec{e}_d$).

\subsection{Discussion}

A standard approach to designing a testing algorithm for a property $\mathcal{P}$ is the following: we identify an alternative property $\mathcal{P}'$ which can be \emph{tested efficiently and exactly}, while satisfying the following:
\begin{itemize*}
	\item[(i)] {\bf Completeness}: If an instance satisfies $\mathcal{P}$, then it satisfies $\mathcal{P}^\prime$.
	\item[(ii)] {\bf Soundness}: If an instance satisfies $\mathcal{P^\prime}$, the it is close to satisfying $\mathcal{P}$.
\end{itemize*}

In other words, we reduce the property testing problem to that of finding a efficiently testable property $\mathcal{P}^\prime$, which can be interpreted as a surrogate for property $\mathcal{P}$. The inherent geometric nature of the problems looked at in this paper motivate us to look for $\mathcal{P}^\prime$'s which are based around convex geometry and high dimensional probability.

For the known design setting, we are looking for a $\mathcal{P}^\prime$, which would ensure that if a given point ${\bf y} \in\R^d$ satisfies $\mathcal{P}^\prime$, then it is close to having a sparse representation in the matrix ${\bf A}$. Towards this end, the approximate Carath\'eodory's theorem states that if a point ${\bf y} \in \R^d$ belonging to the convex-hull of $\vec{A}$, then it is close to another point which admits a sparse representation. On the other hand, if a unit vector ${\bf x} \in \mathcal{S}^{d-1} \cap \mathbbm{R}^d_+$ were $k$-sparse to begin with , then it can be seen that the corresponding ${\bf y} ={\bf A}{\bf x}$ would belong to the convex hull of $\sqrt{k}\cdot \vec{A}$. These observations taken together, seem to suggest that one can take $\mathcal{P}^\prime$ to be membership in the convex-hull of $\sqrt{k}\cdot \vec{A}$. This intuition is made precise in the analysis of the tester in Section \ref{sec:known}.

On other hand, for the unknown design setting identifying the property $\mathcal{P}^\prime$ requires multiple considerations. Here, we are intuitively looking for a $\mathcal{P}^\prime$ based on a quantity $\omega$ that {\em robustly} captures sparsity and is easily computable using linear queries, in the sense that $\omega$ is small when the input vectors have a sparse coding and large when they are ``far" from any sparse coding. Moreover, $\omega$ needs to be invariant with respect to isometries and nearly invariant with respect to near-isometries.
A natural and widely-used measure of structure that satisfies the above mentioned properties is the {\em gaussian width}.
\begin{definition}
	The {\em gaussian width} of  a set $S \subseteq \R^d$ is: 
	$\omega(S) = \E_{\vec{g}}[\sup_{\vec{v} \in S} \langle \vec{g}, \vec{v}\rangle]$
	where $\vec{g} \in \R^d$ is a random vector drawn from $N(0,1)^d$, i.e., a vector of independent standard normal variables.
\end{definition}

The gaussian width of $S$ measures how well on average the vectors in $S$ correlate with a randomly chosen direction. It is invariant under orthogonal transformations of $S$ as the distribution of $\vec{g}$ is spherically symmetric. It is a well-studied quantity in high-dimensional geometry (\cite{Ver15, MendelsonV02}), optimization (\cite{CRPW12, ALMT13}) and statistical learning theory (\cite{BartlettM02}). \ignore{For instance, \cite{CRPW12, ALMT13} show that the gaussian width can be used to give tight characterizations of the number of measurements required to solve linear inverse problems. In the context of learning theory \cite{BartlettM02}, gaussian width (also known as gaussian complexity) of a function class is a well-used complexity measure of a function class and is used to bound the generalization error. It also has connections to other well-studied complexity measures such as VC-Dimension (\cite{MendelsonV02}).} The following bounds are well-known.
\begin{lemma}[See, for example, \cite{RV08, Ver15}]\label{lem:width}
~\\[-2em]
	\begin{enumerate*}
		\item[(i)] If $S$ is a finite subset of $\mathcal{S}^{d-1}$, then $\omega(S) \leq \sqrt{2\log |S|}$.
		\item[(ii)] $\omega(\mathcal{S}^{d-1}) \leq \sqrt{d}$ 
		\item[(iii)] If $S \subseteq \mathcal{S}^{d-1}$ is of dimension $k$, then $\omega(S) \leq \sqrt{k}$.
		\item[(iv)] $\omega(\Sp_k^d) \leq 2\sqrt{3k \log(d/k)}$ when $d/k > 2$ and $k \geq 4$.
	\end{enumerate*}
\end{lemma}

In the context of Theorems \ref{thm:test_rip} and \ref{thm:test_dim}, one can observe that whenever a given set satisfies sparsity or dimensionality constraints, the gaussian width of such sets are small (points (iii) and (iv) from the above Lemma). Therefore, one can hope to test dimensionality or sparsity by computing an empirical estimate of the gaussian width and comparing the estimate to the results in Lemma \ref{lem:width}. While completeness of such testers would follow directly from concentration of measure, establishing soundness would require us to show that approximate converses of points (iii) and (iv) hold as well i.e., whenever the gaussian width of the set $S$ is small, it can be approximated by sets which are approximately sparse in some design matrix (or have low rank). 

For the soundness direction of Theorem \ref{thm:test_rip}, the above arguments are made precise using Lemma \ref{lem:l_infty} and Theorem \ref{thm:sparse_approx}, which show that small gaussian width sets can be approximated by random projections of sparse vectors and vectors with small $\ell_\infty$-norm. For Theorem \ref{thm:test_dim}, we use lemma \ref{lem:appx-rank-bound} which shows that sets with small gaussian width have small approximate rank.

\subsection{Future Work}

Our work opens the possibility of using linear queries to efficiently test other properties of vectors and matrices which arise in machine learning and convex optimization. Some questions directly motivated by this work are:\\

%

\noindent
\textbf{Other notions of distance}: Whether the soundness guarantees of our theorems can be strengthened (especially for the second setting of unknown design matrices) is an interesting direction for future work. In the unknown design setting, can we have that if the tester accepts, then $\vec{Y} = \vec{AX} + \vec{W}$ where columns of $\vec{X}$ are $O(k)$-sparse and the column norms $\|\vec{W}_j\|=O(\eps)$?\\

\noindent
\textbf{Lower bounds}: What is the minimum number of linear queries needed to test sparsity over known and unknown design matrices?  It seems that a mix of information-theoretic and analytic tools will be needed to prove such lower bounds.\\

\noindent
\textbf{Other restrictions on the dictionary}: Another important direction of future work is to consider our two testing problems in the context of commonly used dictionaries, such as the ones composed of Fourier basis, wavelet basis, and ridgelets. In particular, these dictionaries do satisfy RIP, but given their applicability it is relevant to understand if the results obtained in this paper can be strengthened with these additional restrictions on the dictionary. \\

\noindent
\textbf{Construction of $\epsilon$-nets}: A key technical contribution of the paper is to show that an $\epsilon$-net of the unit sphere can be obtained by projecting down (from an appropriately larger dimension) the set of sparse vectors; see Lemma~\ref{thm:sparse_approx}. It might be of independent interest to understand if one can obtain such nets by projecting other structured sets with high gaussian width. 


\section{Preliminaries}

Given $S \subset \mathbbm{R}^d$, we shall use ${\rm conv}(S)$ to denote the convex hull of $S$. For a vector ${\bf x} \in \mathbbm{R}^d$, we use $\|\cdot\|_p$ to denote its $\ell_p$-norm, and we will drop the indexing when $p = 2$. We denote the $\ell_2$-distance of the point ${\bf x}$ to the set $S$ by ${\rm dist}({\bf x},S)$. We recall the definition of $\epsilon$-isometry:

\begin{definition}
	Given sets $S \subset \mathbbm{R}^{m}$ and $S^\prime \subset \mathbbm{R}^n$ (for some $m,n \in \mathbbm{N}$), we say that $S^\prime$ is an $\epsilon$-isometry of $S$, if there exists a mapping $\psi:S \mapsto S^\prime$ which satisfies the following property:
	\begin{equation*}
	\forall {\bf x},{\bf y} \in S : (1 - \epsilon)\|{\bf x} - {\bf y}\| \le \|\psi({\bf x}) - \psi({\bf y})\| \le (1 + \epsilon)\|{\bf x} - {\bf y}\|
	\end{equation*}
\end{definition}

For the unknown design setting, we shall require the notion of Restricted Isometry Property, which is defined as follows:

\begin{definition}[$(\epsilon,k)$-RIP]  
	A matrix ${\bf A} \in \mathbbm{R}^{d \times m}$ satisfies $(\epsilon,k)$-RIP, if for every ${\bf x} \in \Sp^m_k$ the following holds:
	\begin{equation}
	(1 - \epsilon)\|{\bf x}\| \le \|{\bf A}{\bf x}\| \le (1 + \epsilon)\|{\bf x}\| 
	\end{equation}
\end{definition}

We use the following version of Gordon's Theorem repeatedly in this work.
\begin{theorem}[Gordon's Theorem \cite{Gordon1985}]\label{thm:gordon}
	Given $S \subset \mathcal{S}^{D-1}$ and a random gaussian matrix ${\bf G} \sim \frac{1}{\sqrt{d^\prime}}N(0,1)^{d^\prime \times D}$, we have
	\begin{equation*}
		\E_{\bf G}\Big[\max_{{\bf x} \in S} \|{\bf G}{\bf x}\|_2\Big] \le 1 + \frac{\omega(S)}{\sqrt{d^\prime}}
	\end{equation*}
\end{theorem}
It directly implies the following generalization of the Johnson-Lindenstrauss lemma.	
\begin{theorem}[Generalized Johnson-Lindenstrauss lemma]			\label{thm:jl_lemma}
	Let $S \subseteq \mathcal{S}^{n-1}$. Then there exists linear transformation $\Phi:\R^n \mapsto \R^{d^\prime}$, for $d^\prime = O\Big(\frac{\omega(S)^2}{\epsilon^2}\Big)$, such that $\Phi$ is an $\eps$-isometry on $S$. Moreover,
	$\Phi \sim \frac{1}{\sqrt{d'}}N(0,1)^{d^\prime \times n}$ is an $\eps$-isometry on $S$ with high probability. \ignore{More precisely:
		\begin{equation*}
			\Pr_\Phi\left[ \max_{\vec{x}, \vec{y} \in S}  \|\Phi(\vec{x} - \vec{y})\|_2 \ge \left(1 + \big(1 + \epsilon\big)\frac{\omega(S)}{\sqrt{d}}\right)\|\vec{x}-\vec{y}\|_2 \right] \le \exp\Big(-O(\epsilon\omega(S))^2\Big)
		\end{equation*}}
	\end{theorem}
	
	It can be easily verified that the quantity   $\max_{{\bf x} \in S} \|{\bf G}{\bf x}\|_2$ is $1$-Lipschitz with respect to ${\bf G}$. Therefore, using Gaussian concentration for Lipschitz functions, we get the following corollary :
	\begin{corollary}										\label{corr_lipschitz}
		Let $S$ and $G$ be as in Theorem \ref{thm:gordon}. Then for all $\epsilon > 0$, we have
		\begin{equation*}
			\Pr_{\bf G}\bigg( \max_{{\bf x} \in S}  \|{\bf G}{\bf x}\|_2 \ge 1 + \big(1 + \epsilon\big)\frac{\omega(S)}{\sqrt{d^\prime}}\bigg) \le \exp\Big(-O(\epsilon\omega(S))^2\Big)
		\end{equation*}
	\end{corollary}

	\ignore{
		The Gaussian width gives us a way of upper bounding the covering number of the point set.
		
		\begin{theorem}[Sudakov Minoration]				\label{thm:minoration}	
			Let $S$ be a compact subset of $\R^d$. Let $P_\epsilon$ be a minimal $\epsilon$-net over $S$. Then 
			\begin{equation}
			\log|P_\epsilon| \le C\Big(\frac{\omega(S)}{\epsilon}\Big)^2
			\end{equation}
		\end{theorem}
	}
	
	The following lemma gives concentration for the gaussian width:
	\begin{lemma}[Concentration on the gaussian width \cite{boucheron2013concentration}]\label{lem:gconc}
		Let $S \subset \mathbbm{R}^d$. Let $W = \sup_{\vec{v} \in S} \langle \vec{g}, \vec{v}\rangle$ where $\vec{g}$ is drawn from $N(0,1)^d$. Then:
		$$\Pr[|W - \E W| > u] < 2 e^{-\frac{u^2}{2\sigma^2}}$$
		where $\sigma^2 = \sup_{\vec{v} \in S}\big(\|\vec{v}\|^2_2\big)$. Notice that the bound is dimension independent.
	\end{lemma}
	
	We shall also need the following comparison inequality relating suprema of gaussian processes:
	\begin{lemma}[Slepian's lemma \cite{slepian}]\label{lem:slepian}
		Let $\{X_{u}\}_{u \in U}$ and $\{Y_u\}_{u \in U}$ be two almost surely bounded centered Gaussian processes, indexed by the same compact set $U$. If for every $u_1,u_2 \in U$:
		$$\E\bigg[|X_{u_1} -  X_{u_2}|^2\bigg] \le \E\bigg[|Y_{u_1} -  Y_{u_2}|^2\bigg]$$
		then we have 
		$$\E\bigg[\sup_{u \in U} X_u\bigg] \le \E\bigg[\sup_{u \in U} Y_u\bigg]$$	
	\end{lemma}
	
	Lastly, we shall use the $\ell_2$-variant of the approximate Carath\'{e}odory's Theorem:
	
	\begin{theorem}(Theorem $0.1.2$ \cite{Vershynin16} )
		Given $X = \{{\bf w}_1,\ldots,{\bf w}_p\}$ where $\|{\bf w}_i\| \le 1 $ for every $i \in [p]$. Then for every choice ${\bf z} \in {\rm conv}\big(X\big)$ and $k \in \mathbbm{N}$, there exists ${\bf w}_{i_1},{\bf w}_{i_2},\ldots,{\bf w}_{i_k}$ such that 
		\begin{equation}
		\bigg\| \frac{1}{k}\sum_{j \in [k]}{\bf w}_{i_j} - {\bf z} \bigg\| \le \frac{2}{\sqrt{k}} 
		\end{equation} 
	\end{theorem}

\subsection{Algorithmic Estimation of Gaussian Width and Norm of a vector}

We record here simple lemmas bounding the number of linear queries needed to estimate the gaussian width of a set and the length of a vector.

\begin{lemma}[Estimating Gaussian Width using linear queries]			\label{lem:estim}
	For any $u > 4$, $\eps \in (0,1/2)$ and $\delta > 0$, there is a randomized algorithm that given a set $S \subseteq \R^d$ and $\|\vec{v}\| \in [1 \pm \epsilon]$ for all $\vec{v} \in S$, computes $\hat{\omega}$ such that   $\omega(S) - u \leq \hat{\omega} \leq \omega(S) + u$  with probability at least $1-\delta$. The algorithm makes $O(\log(1/\delta) \cdot |S|)$ linear queries to $S$.
\end{lemma}
\begin{proof}
	By Lemma \ref{lem:gconc}, for a random $\vec{g} \sim N(0,1)^d$, $\sup_{\vec{v} \in S} \langle \vec{g}, \vec{v} \rangle$ is away from $\omega(S)$ by $u$ with probability at most $2e^{-16/4.5} < 0.1$. By the Chernoff bound, the median of $O(\log \delta^{-1})$ trials will satisfy the conditions required of $\hat{\omega}$ with probability at least $1-\delta$.
\end{proof}
\vspace{-1em}
\begin{lemma}[Estimating norm using linear queries]			\label{lem:normestim}
	Given $\eps \in (0,1/2)$ and $\delta > 0$, for any vector ${\bf x} \in  \mathbbm{R}^d$ , only $O(\eps^{-2} \log \delta^{-1})$ linear queries to $\vec{x}$ suffice to decide whether $\|\vec{x}\| \in [1-\eps, 1+\eps]$ with success probability $1-\delta$.
	
\end{lemma}
\begin{proof}
	It is easy to verify that ${\bf E}_{{\bf g} \sim N(0,1)^d} [\langle{\bf g},{\bf x}\rangle^2] = \|{\bf x}\|^2$. Therefore, it can be estimated to a multiplicative error of $(1\pm \epsilon/2)$ by taking the average of the squares of linear measurements using $O\Big(\frac{1}{\epsilon^2}\log\frac{1}{\delta}\Big)$-queries. For the case $\|{\bf x}\|_2 \le 2$, a multiplicative error $(1 \pm \epsilon/2)$ implies an additive error of $\epsilon$. Furthermore, when $\|{\bf x}\|_2 \ge 2$, a multiplicative error of $(1 \pm \epsilon/2)$ implies that $L \ge 2(1 - \epsilon/2) > 1 + \epsilon$ for $\eps < 1/2$. 
\end{proof}

\section{Analysis for the Known Design setting}		\label{sec:known}

In this section, we describe and analyze the tester for the known design matrix case. The algorithm itself is a simple convex-hull membership test, which can be solved using a linear program.
 
\RestyleAlgo{boxruled}
\LinesNumbered
\begin{algorithm}
	Set $n = 100k{\log \frac{m}{\delta}}$, sample projection matrix  $\Phi \sim \frac{1}{\sqrt{n}}N(0,1)^{n \times d}$\;
	Observe linear sketch $\tilde{\bf y} = \Phi({\bf y})$\;
	Let $A_\pm = A \cup -A$\; 
	Accept iff $\tilde{\bf y} \in \sqrt{k}\cdot{\rm conv}\big(\Phi(A_\pm)\big)$\;
	\caption{SparseTest-KnownDesign} 
\end{algorithm}       

The guarantees of the above tester are restated in the following Theorem:

\known*

We shall now prove the completeness and soundness guarantees of the above tester. The running time bound follows because convex hull membership reduces to linear programming.

\subsection{Completeness}

Let ${\bf y} = {\bf A}{\bf x}$ where ${\bf A} \in \mathbbm{R}^{d \times m}$ is an arbitrary matrix with $\|{\bf a}_i\| = 1$ for every $i \in [m]$. Furthermore $\|{\bf x}\|_2 = 1$ and $\|{\bf x}\|_0 \le k$. Therefore, by Cauchy-Schwartz we have $\|{\bf x}\|_1 \le \sqrt{k}\|{\bf x}\|_2 = \sqrt{k}$. Hence, it follows that ${\bf y} \in \sqrt{k}\cdot{\rm conv}(A_\pm)$. Since $\Phi:\mathbbm{R}^m \mapsto \mathbbm{R}^d$ is a linear transformation, we have $\Phi({\bf y}) \in \sqrt{k}\cdot{\rm conv}(\Phi(A_\pm))$. Therefore, the tester accepts with probability $1$.

\subsection{Soundness}

Consider the set $A_{\epsilon/\sqrt{k}}$ which is the set of all $(2k/\epsilon^2)$-uniform convex combinations of $\sqrt{k}({A}_\pm)$ i.e.,
\begin{equation}
A_{\epsilon/\sqrt{k}} = \bigg\{ \sum_{{\bf v}_i \in \Omega} \frac{\epsilon^2}{2k}{\bf v}_i : \mbox{multiset } \Omega \in \Big({\sqrt{k}.A_\pm}\Big)^{2k/\epsilon^2}\bigg\}
\end{equation} 

 Then, from the approximate Carath\'{e}odory theorem, it follows that $A_{\epsilon/\sqrt{k}}$ is an $\epsilon$-cover of $\sqrt{k}\cdot{\rm conv}\big(A_\pm\big)$. Furthermore, $|A_{\epsilon/\sqrt{k}}| \le (2m)^{2k/\epsilon^2}$. By our choice of $n$, with probability at least $1 - \delta/2$, the set $\Phi\Big(\{{\bf y}\} \cup A_{\epsilon/\sqrt{k}}\Big)$ is $\epsilon$-isometric to $\{{\bf y}\} \cup A_{\epsilon/\sqrt{k}}$.

Let $\tilde{A}_{\epsilon/\sqrt{k}} = \Phi\big(A_{\epsilon/\sqrt{k}}\big)$. Again, by the approximate Carath\'{e}odory's theorem, the set $\tilde{A}_{\epsilon/\sqrt{k}}$ is an $\epsilon$-cover of $\Phi\big(\sqrt{k}\cdot{\rm conv}(A_\pm)\big)$. Now suppose the test accepts ${\bf y}$ with probability at least $\delta$. Then, with probability at least $\delta/2$, the test accepts and the above $\epsilon$-isometry conditions hold simultaneously. Then,
\begin{eqnarray*}
\tilde{\bf y} \in \sqrt{k}\cdot{\rm conv}\big(\Phi(A_\pm)\big) &\overset{1}{\Rightarrow}& {\rm dist}\big(\tilde{\bf y},\tilde{A}_{\epsilon/\sqrt{k}}\big) \le \epsilon \\
&\overset{2}{\Rightarrow}& {\rm dist}\big({\bf y},{A}_{\epsilon/\sqrt{k}}\big) \le \epsilon(1 - \epsilon)^{-1} \le 2\epsilon \\
&{\Rightarrow}& {\rm dist}\big({\bf y},\sqrt{k}\cdot{\rm conv}(A_\pm)\big) \le 2\epsilon 
\end{eqnarray*}  

 \noindent where step $1$ follows from the $\epsilon$-cover guarantee of $\tilde{A}_{\epsilon/\sqrt{k}}$, step $2$ follows from the $\epsilon$-isometry guarantee. Invoking the approximate Carath\'{e}odory theorem, we get that there exists $\hat{\bf y} = {\bf A}\hat{\bf x}\in \sqrt{k}\cdot{\rm conv}(\pm A)$ such that $\|\hat{\bf x}\|_0 \le O(k/\epsilon^2)$ and $\|\hat{\bf y} - {\bf y}\| \le O(\epsilon)$. This completes the soundness direction.

		\section{Analysis for Unknown Design setting}
		In this section, we restate and prove Theorem \ref{thm:test_rip}. 
		
		\unknown*
				
		Let $S$ denote the set $\{\vec{y}_1, \dots, \vec{y}_p\}$. Our testing algorithm is as follows:
		
		\RestyleAlgo{boxruled}
		\LinesNumbered
		\begin{algorithm}
			Use Lemma \ref{lem:normestim} to decide with probability at least $1-\delta/2$ if there exists $\vec{y}_i$ such that $\|\vec{y}_i\| \not \in [1-2\eps, 1+2\eps]$. Reject if so.\\
			Use Lemma \ref{lem:estim} to obtain $\hat{\omega}$, an estimate of $\omega(S)$ within additive error $\sqrt{3k \log(m/k)}$ with probability at least $1-\delta/2$.\\
			Accept if $\hat{\omega} \le 4\sqrt{3k\log({m}/{k})}$, else reject.
			\caption{SparseTestUnknown} 
		\end{algorithm}
		The number of linear queries made by the tester is $O(p\eps^{-2}\log(p/\delta))$ in  Line 1 and $O(p\log\delta^{-1})$ in Line 2.
		
		\subsection{Completeness}
		Assume that for each $i \in [p]$, $\vec{y}_i = A \vec{x}_i$ for a matrix $A \in \R^{d \times m}$ satisfying $(\eps, k)$-RIP and $\vec{x}_i \in \Sp_k^m$. By definition of RIP, we know that $1-\eps \leq \|\vec{y}_i\| \leq 1+\eps$, so that Line 1 of the algorithm will pass with probability at least $1-\delta/2$.
		
		From Lemma \ref{lem:width}, we know that $\omega(\{\vec{x}_1, \dots \vec{x}_p\}) \leq 2\sqrt{3k\log(m/k)}$.   Lemma \ref{lem:gw_isometricEmbeddings} shows that the gaussian width of $S$ is approximately the same; its proof, deferred to the appendix (Section \ref{sec:gw_isometry}), uses Slepian's Lemma (Lemma \ref{lem:slepian}).
		\begin{lemma}			\label{lem:gw_isometricEmbeddings}
			Let $X \subset \mathcal{S}^{m-1}$ be a finite set, and let $S \subset \R^{d}$ be an $\epsilon$-isometric embedding of $X$. Then 
			\begin{equation}
				(1 - {\epsilon})\omega(X) \le \omega({S}) \le (1 + {\epsilon})\omega(X)
			\end{equation}
		\end{lemma}
		Hence, the gaussian width of $\vec{y}_1, \dots, \vec{y}_p$ is at most $2(1+\eps)\sqrt{3k\log(m/k)}$. Taking into  account the additive error in Line 2, we see that with probability at least $1-\delta/2$, $\hat{\omega} \leq (3 + 2\eps)\sqrt{3k\log(m/k)} \leq 4 \sqrt{3k\log(m/k)}$.  Hence, the tester accepts with probability at least $1-\delta$.

			\subsection{Soundness}
			
			As mentioned before, in order to prove soundness  we need to show that whenever the gaussian width of the set $S$ is small, it is \emph{close} to some sparse point-set. Let $\omega^* = 4\sqrt{3k\log\frac{m}{k}}$. We shall break the analysis into two cases:

		\vspace{1em}	
	   \noindent{\bf Case (i)} $\bigg\{\omega^* \ge (\epsilon/C)^2\sqrt{d}\bigg\}$: For this case, we use the fact random projection of discretized sparse point-sets (Definition \ref{defn:discrete_sparse}) form an appropriated cover of $S$. 	This is formalized in the following theorem, which in a sense shows an approximate inverse of Gordon's Theorem for sparse vectors:
		
		\begin{theorem}					\label{thm:sparse_approx}
		Given $\eps > 0$ and integers $C, d, k$ and $m$,  let $n = O\Big(\frac{k}{\eps^2} \log(m/k)\Big)$. Suppose $m \ge k/\eps^8$. Let $\Phi:\mathbbm{R}^m \mapsto \mathbbm{R}^{n}$ be drawn from $\frac{1}{\sqrt{n}}N(0,1)^{n\times m}$. Then, for $\ell = O({k}{\epsilon^{-4}})$, with high probability, the set $\Phi^{\rm norm}(\widehat{\Sp}_\ell^m)$ is an $O({\epsilon}^{1/4})$-cover of $\mathcal{S}^{n - 1} $,  where $\Phi^{\rm norm}({\bf x}) = \Phi({\bf x})/\|\Phi({\bf x})\|_2$.
		\end{theorem}
		
		From the choice of parameters we have $d \le \frac{C^\prime k}{\epsilon^2}\log\frac{m}{k}$ Therefore, using the above Theorem we know that there exists $(\epsilon,k)$-RIP matrix $\Phi \in \mathbbm{R}^{d \times m}$ such that $\Phi^{\rm norm}\big(\Sp^m_\ell\big)$ is an $O(\epsilon^{1/4})$-cover of $\mathcal{S}^{d-1}$ (and therefore it is a $\epsilon^{1/4}$-cover of $S$). Therefore, there exists ${\bf X} \in \mathbbm{R}^{m \times p} $ such that ${\bf Y} = \Phi({\bf X}) + {\bf E}$ where the columns of ${\bf X}$ and ${\bf E}$ satisfy the respective $\|\cdot\|_0$ and $\|\cdot\|_2$-upper bounds respectively. Hence the claim follows.
		
		\vspace{1em}
		\noindent{\bf Case (ii)} $\bigg\{\omega^* \le (\epsilon/C)^2\sqrt{d}\bigg\}$: For this case, we use the following result on the concentration of $\ell_\infty$-norm:
		
		\begin{lemma}					\label{lem:l_infty}
			Given $S \subset \mathcal{S}^{d-1}$, we have 
			\begin{equation*}
			\Pr_{{\bf R} \sim \mathbbm{O}_d}\Bigg[ \max_{ {\bf y} \in {\bf R}(S) } \|{\bf y}\|_\infty \le C\frac{\omega(S)}{d^{1/2}}\Bigg] \ge \frac12
			\end{equation*}
			where $\mathbbm{O}_d$ is the orthogonal group in $\mathbbm{\bf R}^d$ i.e., ${\bf R}$ is a uniform random rotation.
		\end{lemma}
		
		Although this concentration bound is known, for completeness we give a proof in the appendix (Section \ref{sec:l_infty}). From the above lemma, it follows that there exists ${\bf R} \in \mathbbm{O}_d$ such that for any ${\bf z} \in Z:= {\bf R}(S)$ we have $\|{\bf z}\|_\infty \le \epsilon^2$ and therefore ${\bf Y} = {\bf  R}^{-1}{\bf Z}$. Furthermore, since ${\bf R}$ is orthogonal, therefore the matrix ${\bf R}^{-1}$ is also orthogonal, and therefore it satisfies $(\epsilon,k)$-RIP. 
		
		To complete the proof, we observe that even though the given factorization has inner dimension $d$, we can trivially extend it to one with inner dimension $m$. This can be done by constructing $\Phi = \big[{\bf R}^{-1} \enskip  {\bf G}\big]$ with ${\bf G} \sim \frac{1}{\sqrt{d}}N(0,1)^{d \times {m - d}}$. Since $\omega^* << d$, from Theorem \ref{thm:jl_lemma} it follows that with high probability ${\bf G}$ (and consequently $\Phi$) will satisfy $(\epsilon,k)$-RIP. Finally, we construct $\hat{\bf Z} \in \mathbbm{R}^{m \times n}$ by padding ${\bf Z}$ with $m - d$ rows of zeros. Therefore, by construction ${Y} = \Phi\cdot\hat{\bf Z}$, where for every $i \in [p]$ we have $\|{\bf z}_i\|_\infty \le \epsilon^2$. Hence the claim follows.  
		
		\ignore{	
			{
				\begin{figure}			\label{tab:notations}
					\centering
					\begin{tabular}{c l} \toprule
						\textbf{Quantity} & $\qquad$\textbf{Interpretation} \\ \midrule\midrule
						$d$ & Native dimension of the input \\
						$n$ & Lower dimension as given by the JL lemma \\
						$\epsilon$  & User input for error tolerance   \\
						$\Phi$  & Scaled Gaussian matrix as given by the JL lemma of dimension $n \times d$  \\
						$\Phi^{\rm norm}$ & Normalized analogue of $\Phi$ \\
						$\phi_i$  & $i^{th}$ Block of $\Phi$ of dimensions $n \times Nn$ \\
						$k$  & Sparsity parameter for completeness   \\
						$\ell$  & Sparsity parameter for soundness   \\ 
						$\mathcal{S}^{d-1}$ & Unit sphere in $d$-dimensions \\
						$\Sp^m_k$  & Set of $k$-sparse points in $\mathbbm{R}^m$  \\
						$\widehat{\Sp}^{m}_\ell$  & Discretized set of $\ell$-sparse points in $\mathbbm{R}^m$ \\ \bottomrule
					\end{tabular}
					\caption{Table of Notations}
				\end{figure}
			}
		}

		\section{Proof of Theorem \ref{thm:sparse_approx}}
		
		We begin by defining the discretized sparse point set for $\ell = O(k/\epsilon^4)$:
		
		\begin{definition}[Discretized sparse vectors]					\label{defn:discrete_sparse}
			$$\widehat{\Sp}_\ell^m = \left\{{\bf x} : {\bf x} \in \left \{0,\pm \frac{1}{\sqrt{\ell}}\right \}^m, \|{\bf x}\|_0 = \ell\right\}$$
		\end{definition}
		
		The intent here is to get the sparse point set $\widehat{\Sp}^m_\ell$  \emph{distorted} on projection, so that it forms an $\epsilon$-cover of the unit sphere on the smaller dimension. However, doing so rules out proofs that rely on simple union bound arguments. For instance, on allowing the projections to become distorted, we run into the risk of lots of points collapsing together into a small fraction of $\mathcal{S}^{n-1}$. As a result, the set ${\Phi}^{\rm norm}(\widehat{\Sp}^m_\ell)$ could turn out to be insufficient for forming a cover of the unit sphere. These issues are avoided by carefully relating the gaussian width of ${\Phi}^{\rm norm}(\widehat{\Sp}^m_{\ell})$ to that of $\mathcal{S}^{n-1}$, followed by a partitioning argument. The partitioning crucially uses the block structure of elements in $\widehat{\Sp}^m_\ell$, which results in independent and distributionally identical blocks, allowing us to take union bounds effectively.   
		
		\begin{proof}
			The proof of this Theorem proceeds in steps. We first partition  $\Phi$ into $L$ blocks of $Nn$ columns each, for some appropriately chosen $N$. So, $\Phi = [\phi_1 \cdots \phi_L]$. Note that $\phi_i $ is a $n \times Nn$ submatrix of $\Phi \in \mathbbm{R}^{n \times m}$. Write $\phi^{\rm norm}_i(\widehat{\Sp}_{\ell})$ to denote $\phi^{\rm norm}_i(\widehat{\Sp}_{\ell}^{Nn})$. Also, note that $\phi^{\rm norm}_i(\widehat{\Sp}_{\ell})$ is equal to the set obtained by applying $\Phi$ to vectors in $\widehat{\Sp}_{\ell}$ whose support is contained in $P$, where $P \subset [m]$ is the set of $Nn$ columns of $\Phi$ that are present in $\phi_i$. \\
			
			For any such fixed partition $\phi_i$, we show that the restriction $\phi^{\rm norm}_i(\widehat{\Sp}_{\ell})$ has large expected gaussian width (Lemma \ref{lem:gw_lowerBound}), where again, $\phi^{\rm norm}_i(\vec{x}) = \phi_i(\vec{x})/\|\phi_i(\vec{x})\|$.  Furthermore, using Lemma \ref{lem:min_dist} we argue that any fixed point on $\mathcal{S}^{n - 1}$ has distance $O(\epsilon^{1/4})$ to $\phi_i(\widehat{\Sp}_{\ell})$ with large probability. Now, we  use the independence of the $\phi_i$'s to argue that the probability of ${\bf x} \in \mathcal{S}^{n - 1}$ being simultaneously far away from all $\phi_i^{\rm norm}(\widehat{\Sp}_{\ell})$ is exponentially small. Finally, taking a union bound over the $\epsilon$-net of $\mathcal{S}^{n - 1}$ completes the proof.
			
			We fix the parameter $N$ to be $N = \sqrt{\frac{m}{k}}\Big(\log\frac{m}{k}\Big)^{-1}$.		Set the number of blocks $L = \frac{m}{Nn} = O(\epsilon^2\sqrt{m/k})$ blocks of $Nn$ coordinates each. By construction, for any fixed $i \in [L]$, $\phi_i \sim \frac{1}{\sqrt{n}} N(0,1)^{n \times Nn}$. The following lemma allows gives us a lower bound on the gaussian width of each projection:
			
			\begin{lemma}				\label{lem:gw_lowerBound}
				Let $\phi \sim \frac{1}{\sqrt{n}}\mathcal{N}(0,1)^{n \times Nn}$. Then for $\widehat{\Sp}_{\ell} \subset \mathcal{S}^{Nn - 1}$,
				\begin{equation}
					\E_{\phi}\big[\omega(\phi^{\rm norm}(\widehat{\Sp}_{\ell}))\big] \ge (1 - 4\epsilon)\sqrt{n}
				\end{equation}
			\end{lemma}	
			
			From Lemma \ref{lem:gw_lowerBound}, we show the following lower bound on the gaussian width of the projections of $\widehat{\Sp}_\ell^{Nn-1}$ : 
			$\mathbbm{E}_{\phi_i}\Big[ \omega(\phi^{\rm norm}_i(\widehat{\Sp}_{\ell}))\Big] \ge (1 - 4 \epsilon)\sqrt{n}$.
			
			Now, we argue that because $\phi^{\rm norm}_i(\widehat{\Sp}_\ell)$ has gaussian width close to $\sqrt{n}$, it in fact covers any fixed point on $\mathcal{S}^{n-1}$ with high probability. We begin by stating the following lemma on concentration of minimum distance with respect to large gaussian width sets.

			\begin{lemma}					\label{lem:min_dist}
				Let $T \subset \mathcal{S}^{n -1 }$ be such that $\omega(T) \ge \sqrt{n}(1 - \epsilon)$, where $n = O\Big(\frac{k}{\epsilon^2}\log\frac{m}{k}\Big)$. Let ${R}:\R^{n} \mapsto \R^{n}$ be a uniform random rotation operator. Then for any ${\bf x} \in \mathcal{S}^{n - 1}$
				\begin{equation*}
					\Pr_R\left[ \min_{{\bf y} \in R(T)} \| {\bf x} - {\bf y} \|_2 \ge 2\epsilon^{1/4} \right] \le \exp\Big(-O\Big(\frac{k}{\epsilon}\log\frac{d}{k}\Big)\Big)
				\end{equation*}
			\end{lemma}

			Observe that the distribution of $\phi_i^{\rm norm}$ is rotation-invariant; see \ref{lem:rotational_invariance} for a formal proof. Fixing ${\bf x} \in \mathcal{S}^{n - 1}$, we invoke Lemma \ref{lem:min_dist} in our setting:
			\begin{equation}
				\Pr_{\phi_i}\bigg[ \min_{{\bf y} \in \phi^{\rm norm}_i(\widehat{\Sp}_{\ell})} \| {\bf y} - {\bf x} \|_2 > 16\epsilon^{\frac{1}{4}}\bigg] \le \exp\Big(-O\Big(\frac{k}{\epsilon}\log\frac{m}{k}\Big)\Big)
			\end{equation}
			
			Let $P_\epsilon$ be an $\epsilon$-cover of $\mathcal{S}^{n - 1}$ such that $|P_\epsilon| = O(1/\epsilon)^{n}$. Then:
			\begin{align*}
				\Pr_{\Phi}\Bigg[\exists {\bf x} \in P_\epsilon : \min_{{\bf y} \in \Phi^{\rm norm}(\widehat{\Sp}_{\ell})} \|{\bf y} - {\bf x} \|_2 \ge 16{\epsilon}^{\frac{1}{4}}\Bigg] 
				&\le |P_\epsilon| \Pr_{\Phi}\Bigg[\forall \quad i \in [L] \min_{{\bf y} \in \phi^{\rm norm}_i(\widehat{\Sp}_{\ell})} \|{\bf y} - {\bf x} \|_2 \ge 16{\epsilon}^{\frac{1}{4}}\Bigg] \\
				&\overset{1}{=} |P_\epsilon|\prod_{i = 1}^{L}\Pr_{\phi_i}\Bigg[\min_{{\bf y} \in \phi^{\rm norm}_i(\widehat{\Sp}_{\ell})} \|{\bf y} - {\bf x} \|_2 \ge 16{\epsilon}^{\frac{1}{4}}\Bigg] \\
				&= |P_\epsilon|\Bigg(\Pr_{\phi_i}\Bigg[\min_{{\bf y} \in \phi^{\rm norm}_i(\widehat{\Sp}_{\ell})} \|{\bf y} - {\bf x} \|_2 \ge 16{\epsilon}^{\frac{1}{4}}\Bigg]\Bigg)^L \\
				&\le |P_\epsilon| \exp\Big(-O\Big(\frac{k}{\epsilon}\log\frac{m}{k}\Big)\Big)^{\epsilon^2O(\sqrt{m/k})} \\
				&\le \exp\Big(\Big(\log{\frac{C^\prime}{\epsilon}}\Big)\Big(\frac{k}{\epsilon^2}\log\frac{m}{k}\Big)-{k}{\epsilon}\sqrt{\frac{m}{k}}{\log\frac{m}{k}}\Big) \\
				&\le \exp\Big(-O\Big(k\log\frac{m}{k}\big)\Big) 
			\end{align*}
			
			where step $1$ uses the independence of the $\phi_i$'s and the last step uses the fact that $\epsilon \ge \Big(\frac{k}{m}\Big)^{\frac{1}{8}}$. The above inequality states that $\Phi^{\rm norm}(\widehat{\Sp}_{\ell})$ is an $\big(16\epsilon^{1/4}+\epsilon\big)$-cover of $\mathcal{S}^{n - 1}$ with positive probability, which completes the proof.	
			
		\end{proof}

		\subsection{Proof of Lemma \ref{lem:gw_lowerBound}}
		
		The proof of the lemma proceeds in two steps: we first restrict to the case where the maximum $\|\cdot\|_2$ length of the projected vectors is not much larger than the expected value. This is done by using {Gordon's theorem} and Lipschitz concentration for gaussians (see Theorem \ref{thm:gordon} and Corollary \ref{corr_lipschitz}). Following that, we observe that conditioning the expectation by this high probability event on the length of the projected vectors does not affect the expectation by much (Lemma \ref{lem:conditional_lb}). The rest of the proof follows using standard estimates of gaussian widths of $\widehat{\Sp}_{\ell}$ (see Lemma \ref{lem:gwLowerBound_discrete}).
		
		\noindent
		\textbf{Upper Bound on the $\|\cdot\|_2$ length}: By setting $D = Nn$ in Theorem \ref{thm:gordon}, we upper bound the expected maximum  $\|\cdot\|_2$ length : 
		
		\begin{equation}
			\E_{\phi} \Bigg[\max_{{\bf x} \in \widehat{\Sp}_{\ell}} \|\phi({\bf x}) \|_2\Bigg] \le 1 + \frac{\omega(\widehat{\Sp}_{\ell})}{\sqrt{n}}
		\end{equation}
		
		Furthermore, from Lemma \ref{lem:gwLowerBound_discrete}, we have $\omega(\widehat{\Sp}_{\ell}) = \sqrt{C_0\ell\log\frac{Nn}{\ell}}$ for some constant $C_0 > 0$. Therefore by our choice of parameters we have:
		
		\begin{equation}			\label{fact:upper_bound}
			\frac{\omega(\widehat{\Sp}_{\ell})}{\sqrt{n}}  =  \Big(\frac{Ck}{\epsilon^2}\log\frac{m}{k}\Big)^{-\frac{1}{2}}\sqrt{C_0\ell\log\frac{Nn}{\ell}}	=  \frac{C^\prime}{\epsilon}			
		\end{equation}
		
		Note that Eq. \ref{fact:upper_bound} holds with an equality for some constant $C^\prime$. 
		Now consider the event $\mathcal{E}$ where the maximum $\|\cdot\|_2$ length is at most $1+\frac{C'}{\eps}(1+\eps)$. Then by using gaussian concentration for Lipschitz functions (Corollary \ref{corr_lipschitz}), we upper bound probability of the event $\neg\mathcal{E}$:
		
		\begin{eqnarray}
			\Pr_{\phi \sim \frac{1}{\sqrt{n}}N(0,1)^{n \times Nn}} \Bigg(\max_{{\bf x} \in \widehat{\Sp}_{\ell}} \|\phi({\bf x}) \|_2 \ge 1+\frac{C'}{\eps}(1+\eps)\Bigg) 	&\le& \exp\Big( - c\epsilon^2\omega(\widehat{\Sp}_{\ell})^2\Big) \nonumber \hspace{1cm}(\mbox{Using Corollary} \ref{corr_lipschitz})	 \\
			&\le& \exp\Big( - O\Big(\frac{k}{\epsilon^2}\log\frac{m}{k}\Big)\Big) = \epsilon_{k,m}			\label{fact:err_bound1}
		\end{eqnarray}
		where $\epsilon_{k,m}$ can be made arbitrarily small by choosing $k$ large enough. \\

		\noindent	
		\textbf{Lower Bounding the gaussian width} : Recall that the operator $\phi^{\rm norm}$ is defined as $\phi^{\rm norm}({\bf x}) \overset{\rm def}{=} \phi({\bf x})/\|{\phi}(\bf x)\|_2$. The operational expression for the gaussian width of the projected set restricted to coordinates in $[Nn]$ is given by :
		
		\begin{equation}
			\omega(\phi^{\rm norm}(\widehat{\Sp}_{\ell})) = \E_{{\bf g} \sim N(0,1)^{n}}\Bigg[\max_{{\bf x } \in \widehat{\Sp}_{\ell}} {\bf g}^\top\phi^{\rm norm}({\bf x})\Bigg]
		\end{equation}
		
		We shall also need the following lemma which states that conditioning by the large probability event $\mathcal{E}$ does not reduce the expectation by much.
		
		\begin{lemma}						\label{lem:conditional_lb}
			There exists universal constants $\ell_0,m_0$ such that for all $ m \ge m_0$ and $ d \ge k \ge \ell_0$ and the event $\mathcal{E}$ defined as above, we have
			\begin{equation*}
				\E_{{\bf g} \sim N(0,1)^{n},\phi} \Bigg[ \max_{{\bf x} \in \widehat{\Sp}_{\ell} } {\bf g}^{\top}\phi({\bf x}) \Big| \mathcal{E} \Bigg] \ge \E_{{\bf g} \sim N(0,1)^{n},\phi} \Bigg[ \max_{{\bf x} \in \widehat{\Sp}_{\ell} } {\bf g}^{\top}\phi({\bf x})\Bigg] - \gamma_{k,m} 
			\end{equation*} 
			where $\gamma_{k,m}$ decays exponentially in $k$.
			
		\end{lemma}
		
		We defer the proof of the Lemma to Section \ref{sec:conditional_lb}. Equipped with the above, we proceed to lower bound the expected Gaussian width:

		\begin{align}
			\E_{\phi} &\E_{{\bf g} \sim N(0,1)^{n}}\Bigg[\max_{{\bf x } \in \widehat{\Sp}_{\ell}} {\bf g}^\top\phi^{\rm norm}({\bf x})\Bigg] 		\nonumber\\
			\ge& (1 - \epsilon_{k,m})\E_{\phi} \E_{{\bf g} \sim N(0,1)^{n}}\Bigg[ \max_{{\bf x } \in \widehat{\Sp}_{\ell}} {\bf g}^\top\phi^{\rm norm}({\bf x}) \quad\bigg| \quad \mathcal{E} \Bigg]   \nonumber\\
			\ge& \frac{(1 - \epsilon_{k,m})}{1+\frac{C'}{\eps}(1+\eps)}\E_{\phi} \E_{{\bf g} \sim N(0,1)^{n}}\Bigg[ \max_{{\bf x } \in \widehat{\Sp}_{\ell}}  {\bf g}^\top\phi({\bf x}) \quad\bigg|\quad \mathcal{E} \Bigg]  \nonumber\\
			\overset{1}{\simeq}& \frac{(1 - \epsilon_{k,m})}{1+\frac{C'}{\eps}(1+\eps)}\E_{\phi} \E_{{\bf g} \sim N(0,1)^{n}}\Bigg[ \max_{{\bf x } \in \widehat{\Sp}_{\ell}} {\bf g}^\top\phi({\bf x}) \Bigg]     \nonumber
		\end{align}
		where the first inequality follows from the fact that $\mathcal{E}$ is a large probability event, and step $1$ follows from Lemma \ref{lem:conditional_lb}. Removing the conditioning allows us to relate the expectation term to the gaussian width of $\widehat{\Sp}_{\ell}$. Let $B$ denote the event that $\|{\bf g}\|_2 \in [\sqrt{n}(1 - \epsilon/4), \sqrt{n}(1 + \epsilon/4)]$. Using concentration of $\chi^2$-variables, we get $\Pr(B) \ge 1 - \exp\Big(- 4\epsilon^2n\Big) \ge 1 - \epsilon$ since $k \ge C \log\frac{1}{\epsilon}$. Then,
		
		\begin{align}
			\E_{\phi} \E_{{\bf g} \sim N(0,1)^{n}}\Bigg[ \max_{{\bf x } \in \widehat{\Sp}_{\ell}} {\bf g}^\top\phi({\bf x}) \Bigg]
			\ge& (1-\epsilon)\E_{\phi}\E_{{\bf g} \sim N(0,1)^{n}}\Bigg[ \max_{{\bf x } \in \widehat{\Sp}_{\ell}} {\bf g}^\top\phi({\bf x})  \bigg| B \Bigg]   \nonumber\\
			\overset{2}{=}& (1-\epsilon)\E_{{\bf g} \sim N(0,1)^{n}} \E_{ \tilde{\bf g} \sim N\bigg(0,\frac{\|{\bf g}\|^2}{{n}}\bigg)^{Nn}}\Bigg[ \max_{{\bf x } \in \widehat{\Sp}_{\ell}} \tilde{\bf g}^\top{\bf x} \bigg| B\Bigg]  \label{pf:step1}\\
			\overset{3}{\ge}& (1-\epsilon)\E_{{\bf g} \sim N(0,1)^{n}} \E_{ \tilde{\bf g} \sim N\big(0,(1-\epsilon/4)^2\big)^{Nn}}\Bigg[ \max_{{\bf x } \in \widehat{\Sp}_{\ell}} \tilde{\bf g}^\top{\bf x} \bigg| B \Bigg]   \nonumber\\
			=& (1-\epsilon) \E_{ \tilde{\bf g} \sim N\big(0,(1-\epsilon/4)^2\big)^{Nn}}\Bigg[ \max_{{\bf x } \in \widehat{\Sp}_{\ell}} \tilde{\bf g}^\top{\bf x} \Bigg]   \nonumber\\
			\ge& (1-\epsilon)^2\E_{ \tilde{\bf g} \sim N(0,1)^{Nn}}\Bigg[ \max_{{\bf x } \in \widehat{\Sp}_{\ell}} \tilde{\bf g}^\top{\bf x}\Bigg] \label{pf:step2}\\
			\ge& (1 - \epsilon)^2\omega(\widehat{\Sp}_{\ell})   \nonumber
		\end{align}
		In step $2$, in the inner expectation ${\bf g} \in \R^{n}$ is a fixed vector, and therefore ${\bf g}^T{\phi}$ is distributionally equivalent to a gaussian vector in $\R^{Nn}$, scaled by $\frac{\|{\bf g}\|_2}{\sqrt{n}}$ (since the columns of $\phi$ are independent $N(0,1/n)^{Nn}$-gaussian vectors). Step $3$ follows from the lower bound on the $\|\cdot\|_2$-length of ${\bf g}$. Plugging in the lower bound on the expectation term, we get :
		
		\begin{align*}
			\E_{\phi} \E_{{\bf g} \sim N(0,1)^{n}}\Bigg[\max_{{\bf x } \in \widehat{\Sp}_{\ell}} {\bf g}^\top\phi^{\rm norm}({\bf x})\Bigg] 
			&\ge  (1 - \epsilon)^2\frac{(1 - \epsilon_{k,m})}{1+\frac{C'}{\eps}(1+\eps)} \omega(\widehat{\Sp}_{\ell}) \nonumber\\
			&=  (1 - \epsilon)^2\frac{(1 - \epsilon_{k,m})}{1+\frac{C'}{\eps}(1+\eps)} \Big(\frac{C^\prime}{\epsilon} \sqrt{n}\Big) \nonumber\\
			&\ge\sqrt{n}(1 - 4 \epsilon) 	\nonumber
		\end{align*}
		for sufficiently small $\eps$ and large $k$.	
		
		\subsection{Proof of Lemma \ref{lem:min_dist}}

		We begin by looking at the expression of the square of the $\|\cdot\|_2$ distance. For any fixed ${\bf x},{\bf y} \in \mathcal{S}^{n-1}$, we have
		
		\begin{equation}			\label{fact:l2_dist}
			\|{\bf y} - {\bf x} \|^2_2 = 2 - 2{\bf x}^\top{\bf y}
		\end{equation}
		
		Therefore, minimizing the $\|\cdot\|_2$ norm would be equivalent to maximizing the dot product term. Furthermore, it is known that a random gaussian vector ${\bf g} \sim N(0,1)^{n}$ can be rewritten as $Z{\bf r}$ where ${\bf r} \overset{\rm unif}{\sim} \mathcal{S}^{n - 1 }$ and $Z^2$ is a $\chi^2$-random variable with $n$ degrees of freedom. Using this decomposition, we get  :
		
		\begin{align}
			&\Pr_{{\bf g} \sim N(0,1)^{n}}\bigg( \max_{{\bf y} \in T} {\bf g}^\top{\bf y} \le \sqrt{n}(1 - \epsilon - \sqrt{\epsilon}) \bigg)  \nonumber \\
			=& \Pr_{{\bf r},Z}\bigg( Z \max_{{\bf y} \in T} {\bf r}^\top{\bf y} \le \sqrt{n}(1 - \epsilon - \sqrt{\epsilon}) \bigg) \nonumber\\
			\ge &  \Pr_{{\bf r},Z}\bigg( Z \max_{{\bf y} \in T} {\bf r}^\top{\bf y} \le \sqrt{n}(1 - \epsilon - \sqrt{\epsilon}) \Big| Z \le \sqrt{n}(1 + \epsilon) \bigg)\Pr\bigg( Z \le \sqrt{n}(1 + \epsilon) \bigg) \nonumber\\ 	
			\overset{1}{\ge} &  \Pr_{{\bf r},Z}\bigg( Z \max_{{\bf y} \in T} {\bf r}^\top{\bf y} \le \sqrt{n}(1 - \epsilon - \sqrt{\epsilon}) \Big| Z \le \sqrt{n}(1 + \epsilon) \bigg)( 1 - \epsilon) \nonumber\\
			\ge &  \frac{1}{2}\Pr_{{\bf r}}\bigg(  \max_{{\bf y} \in T} {\bf r}^\top{\bf y} \le \frac{1 - \epsilon - \sqrt{\epsilon}}{1 + \epsilon} \bigg) \nonumber\\
			\ge &  \frac{1}{2}\Pr_{{\bf r}}\bigg(  \max_{{\bf y} \in T} {\bf r}^\top{\bf y} \le (1-2\sqrt{\epsilon}) \bigg) \nonumber
		\end{align}
		
		\noindent where in step $1$, we used concentration for $\chi^2$-random variables and use the fact that $k \ge C\log\frac{1}{\epsilon}$. We now relate the behavior of the maximum dot product of a set with respect to a random vector to the maximum dot product of a fixed vector with respect to a randomly rotated set.	
		\begin{align}	
			\Pr_{{\bf r}}\bigg(  \max_{{\bf y} \in T} {\bf r}^\top{\bf y} \le (1-2\sqrt{\epsilon}) \bigg)
			\overset{2}{=} &  \Pr_{R}\bigg(  \max_{{\bf y} \in T} R({\bf x})^\top{\bf y} \le (1-2\sqrt{\epsilon}) \bigg) \nonumber\\
			= &  \Pr_{R}\bigg(  \max_{{\bf y} \in T} {\bf x}^\top R^{-1}({\bf y}) \le (1-2\sqrt{\epsilon}) \bigg) \nonumber\\
			\overset{3}{=} &  \Pr_{R}\bigg(  \max_{{\bf y} \in T} {\bf x}^\top R({\bf y}) \le (1-2\sqrt{\epsilon}) \bigg) \nonumber\\
			= &  \Pr_{R}\bigg(  \max_{{\bf y} \in R(T)} {\bf x}^\top{\bf y} \le (1-2\sqrt{\epsilon}) \bigg) 		\label{fact:gauss_rotationEquivalence}	
		\end{align}
		
		Step $2$ follows from the fact that applying a uniformly random rotation on a unit vector is equivalent to sampling uniformly from the unit sphere $\mathcal{S}^{n - 1}$, and step $3$ follows from the fact that if $R$ is uniformly random rotation, then $R^{-1}$ is also a uniformly random rotation. Furthermore, using gaussian concentration for Lipschitz functions :
		
		\begin{equation}					\label{fact:gw_concentration}
			\Pr_{{\bf g} \sim N(0,1)^{n}}\Big( \max_{{\bf y} \in T} {\bf g}^\top{\bf y} \le \sqrt{n}(1 - \epsilon - \sqrt{\epsilon}) \Big) \le \exp\Big( -O\Big(\frac{k}{\epsilon}\log\frac{m}{k}\Big)\Big)
		\end{equation}\textbf{}
		
		\noindent and the l.h.s is an upper bound on  $\frac{1}{2}\Pr_{R}\Big(  \max_{{\bf y} \in R(T)} {\bf x}^\top{\bf y} \le (1-2\sqrt{\epsilon}) \Big)$ (Eq. \ref{fact:gauss_rotationEquivalence}). Therefore rearranging the equations, we have:
		
		\begin{equation}
			\Pr_{R}\Big(  \min_{{\bf y} \in R(T)} \| {\bf x} - {\bf y}\| \ge 2\epsilon^{1/4} \Big) = \Pr_{R}\Big(  \max_{{\bf y} \in R(T)} {\bf x}^\top{\bf y} \le (1-2\sqrt{\epsilon}) \Big) \le 2 \exp\Big( -O\Big(\frac{k}{\epsilon}\log\frac{m}{k}\Big)\Big)
		\end{equation}

\section{Tolerant testers for Known and Unknown Designs}
		
		The simplicity of the testers for the known and unknown design settings directly translates to their robustness to noise. In this section, we state and prove our results for the tolerant variants of these problems. 
		
		\begin{theorem}				\label{thm:noise_known}
			Fix $\eps \in (0,1)$ and positive integers $d, k, m$ and a matrix ${\bf A} \in \mathbbm{R}^{d \times m}$ such that $\|{\bf a}_i\| = 1$ for every $i \in [m]$. There exists a randomized testing algorithm which makes linear queries to the input vector $\vec{y} \in \R^d$  and has the following properties:
			\begin{itemize}
				\item \emph{\textbf{Completeness}}: If ${\bf y} = {\bf A}{\bf x} + {\bf e}$ for some ${\bf x} \in \Sp_k^m$ such that $\|{\bf e}\| \le \epsilon$, then the tester accepts with probability $\ge 1 - \delta$.
				
				\item \emph{\textbf{Soundness}}: If $\|{\bf A}{\bf x} - {\bf y}\| > \eps$ for every $\vec{x} : \|\vec{x}\|_0 \le K$, then the tester rejects with probability $\ge 1 - \delta$. Here, $K = O(k/\eps^2)$.
			\end{itemize}
			The query complexity of the tester is $O(k \eps^{-2}\log\frac{m}{\delta})$.
		\end{theorem}
			
		The tolerant testing algorithm is the following:
		
		\RestyleAlgo{boxruled}
		\LinesNumbered
		\begin{algorithm}
			Set $n = \frac{200k}{\epsilon^2}\log\frac{m}{\delta}$, sample projection matrix  $\Phi \sim \frac{1}{\sqrt{n}}N(0,1)^{n \times d}$\;
			Observe linear sketch $\tilde{\bf y} = \Phi({\bf y})$\;
			Let $A_\pm = A \cup -A$\; 
			Accept iff ${\rm dist}\Big(\tilde{\bf y},\sqrt{k}.{\rm conv}\big(\Phi(A_\pm)\big)\Big) \le 2\epsilon$\;
			\caption{SparseTestKnown-Noisy} 
		\end{algorithm}       
		
		The difference here is in the final step, where instead of checking exact membership of the point $\tilde{\bf y}$ inside the convex hull, we check if the point is close enough to it. We now prove Theorem \ref{thm:noise_known}:

		\begin{proof}
			We again consider the set $A_{\epsilon/\sqrt{k}}$ from the soundness analysis of Theorem \ref{thm:test_known}. As before, by our choice of $n$, with probability at least $1 - \delta/2$, the set $\Phi\Big(\{{\bf y}\} \cup A_{\epsilon/\sqrt{k}}\Big)$ is $\epsilon$-isometric to $\{{\bf y}\} \cup A_{\epsilon/\sqrt{k}}$. Given this observation, for completeness we observe that 
			\begin{eqnarray*}
				{\bf y} = {\bf A}{\bf x} + {\bf e} &\Rightarrow& {\rm dist}\Big({\bf y},\sqrt{k}\cdot{\rm conv}\big(A_\pm\big)\Big) \le \epsilon \\
				 &\Rightarrow& {\rm dist}\Big({\bf y},A_{\epsilon/\sqrt{k}}\Big) \le 2\epsilon \\
				&\overset{1}{\Rightarrow}& {\rm dist}\Big(\tilde{\bf y},\sqrt{k}.{\rm conv}\big(\Phi(A_\pm)\big)\Big) \le \epsilon(1+\epsilon) \le 2\epsilon
		    \end{eqnarray*}   
			
			\noindent where $1$ follows from the $\epsilon$-isometry guarantee, and hence the tester accepts. The arguments for the soundness direction are identical to the ones used in Theorem \ref{thm:test_known}, and hence the claim follows. 
		\end{proof}

		The noise-tolerant algorithm for testing sparsity in the unknown design setting is the same as the one for Theorem \ref{thm:test_rip}. Hence, we just state and prove the guarantees in the noisy setting:
						
		\begin{theorem}[Testing Noisy Sparse representations]	\label{thm:test_rip_noisy8ye}
			Fix $\eps,\eta, \delta \in (0,1)$ and positive integers $d, k, m$ and $p$, such that $(k/d)^{1/8} < \eps < \frac{1}{100}, k \geq C^\prime\log\frac{1}{\epsilon} $ and $m \geq 20k\eps^{-4},\eta \le (\sqrt{2}-1)\frac{\omega(\Sp_k^m)}{\sqrt{\log p}}$.
			There exists a randomized testing algorithm which makes linear queries to input vectors $\hat{\vec{y}}_1, \hat{\vec{y}}_2, \dots, \hat{\vec{y}}_p \in \R^d$  and has the following properties (where $\widehat{\bf Y}$ is the matrix having $\hat{\vec{y}}_1, \hat{\vec{y}}_2, \dots, \hat{\vec{y}}_p$ as columns):
			\begin{itemize}
				\item \emph{\textbf{Completeness}}: If there exists $Y$ with $\|{\bf y}_i - \hat{\bf y}_i\| \le \eta$ for every $i \in [p]$ and  $Y = AX$ such for some $(\epsilon,k)$-RIP matrix $A \in \R^{d \times m}$ and $X \in \R^{m \times p}$ with each column of $X$ in $\Sp_k^m$, then the tester accepts with probability $\ge 1 - \delta$.
				
				\item \emph{\textbf{Soundness}}: If ${\bf Y}$ does not admit factorization ${\bf Y} = {\bf A}({\bf X} + {\bf Z}) + {\bf W}$ with 
				\begin{itemize}
					\item[1.] The design matrix ${\bf A} \in \mathbbm{R}^{d \times m}$ being $(\epsilon,k)$-RIP, with $\|{\bf a}_i\| = 1$ for every $i \in [m]$
					\item[2.] The coefficient matrix ${\bf X} \in \mathbbm{R}^{m \times p}$ being column wise $\ell$-sparse, where $\ell = O(k/\eps^4)$.
					\item[3.] The error matrices ${\bf Z} \in \mathbbm{R}^{m \times p}$ and ${\bf W} \in \mathbbm{R}^{d \times p}$ satisfying
					\begin{equation*}
						\|{\bf z}_i\|_\infty \le \epsilon^2, \qquad \|{\bf w}_i\|_2 \le O(\epsilon^{1/4}) \qquad \text{for all } i \in [p]. 
					\end{equation*}
				\end{itemize}	
				Then the tester rejects with probability $\ge 1 - \delta$. 
			\end{itemize}
			The query complexity of the tester is $O(\epsilon^{-2}\log{(p/\delta)})$.
		\end{theorem}
		\begin{proof}
			
			For the completeness, let there be a $Y \in \mathbbm{R}^{d \times n}$ s.t. $d_{\mathcal{H}}(Y,\hat{Y}) \le \eta$ i.e., $Y$ is column wise $\eta$-close to $\hat{Y}$ in the $\ell_2$-norm, as in the completeness criteria. We can then upper bound the gaussian width of the perturbed set as :  
			\begin{eqnarray*}
				\omega(\widehat{Y}) = \E_{\bf g}\bigg[ \max_{\hat{\bf y} \in \hat{Y}} {\bf g}^\top\hat{\vec{y}} \bigg] 
				&=& \E_{\bf g}\bigg[ \max_{\hat{\bf y} \in \hat{Y}} {\bf g}^\top{\vec{y}} + {\bf g}^\top\big(\hat{\vec{y}} - \vec{y}\big) \bigg] \\
				&\le& \E_{\bf g}\bigg[ \max_{\hat{\bf y} \in \hat{Y}} {\bf g}^\top{\vec{y}}\bigg] +\E_{\bf g}\bigg[{\bf g}^\top\big(\hat{\vec{y}} - \vec{y}\big) \bigg] \\
				&\overset{1}{\le}& \E_{\bf g}\bigg[ \max_{{\bf y} \in {Y}} {\bf g}^\top{\vec{y}}\bigg]  + C\eta\sqrt{\log p}\\
				&\le& \omega(\Sp_{2k}) 
			\end{eqnarray*}	
			\noindent where step $1$ follows from the observation that the maximum of $p$ (not-necessarily i.i.d) gaussians (with variance at most $\eta^2$) is upper bounded by $O(\eta{\sqrt{\log p}})$, and the last step follows from our choice of $\eta$. 	Now as in Theorem \ref{thm:test_rip}, with high probability the gaussian width estimated by the tester is at most $\omega(\Sp^m_{4k})$ and therefore the tester accepts. For the soundness, if the tester accepts with high probability then $\omega(\widehat{Y}) \le \omega(\Sp^m_{6k}) $, and therefore soundness follows using arguments identical to the main theorem.
		\end{proof}	
		
		\begin{remark}
			Note that the noise model being considered here is adversarial as opposed to the standard gaussian noise. This is a relatively stronger assumption in the sense that an adversary can perturb the vectors depending on the instance i.e., the noise here can be worst case.
		\end{remark}

\section*{Acknowledgements}
We would like to thank David Woodruff for showing us the sketching-based tester described in Section \ref{sec:relwork}.

\bibliographystyle{alpha}
\bibliography{testing}
\appendix
\section{Gaussian Width of the discretized sparse set $\widehat{\Sp}_k$}

\begin{lemma}								\label{lem:gwLowerBound_discrete}
	Let $\widehat{\Sp}_{\ell} \subset \mathcal{S}^{m-1}$ be the discrete $k$-sparse set on the unit sphere. Then,
	\begin{equation}
	\omega(\widehat{\Sp}_{\ell}) = \Theta\Big(\sqrt{\ell\log\frac{m}{\ell}}\Big)
	\end{equation}
	
\end{lemma}

\begin{proof}
	
	For the upper bound, observe that $\widehat{\Sp}_{\ell} \subset {\Sp}^m_{\ell}$ and gaussian width is monotonic. Therefore, from Lemma \ref{lem:width}, we have $\omega(\widehat{\Sp}_{\ell}) = O\Big(\sqrt{\ell\log\frac{m}{\ell}}\Big)$. Towards proving the asymptotic lower bound:  Given independent gaussians $g_1,\ldots,g_n \sim N(0,1)$, it is known that 
	
	\begin{equation}				\label{fact:gauss_block}
	\E_{g_1,\ldots,g_n}\Big[\max_{g_i} g_i\Big] \ge C_0\sqrt{\log n}
	\end{equation}
	
	for some constant $C_0>0$ independent of the number of gaussians.   Now, without loss of generality let $m$ be divisible by $\ell$. We partition the $m$ coordinates into $\ell$ blocks $B_1,\ldots,B_\ell$ of $\frac{m}{\ell}$ coordinates each. Then,
	
	\begin{equation}
	\E_{\bf g}\Big[ \max_{{\bf x} \in \widehat{\Sp}_{\ell}} {\bf g}^\top{\bf x}\Big] \ge \frac{1}{\sqrt{\ell}}\E \Big[\sum_{j \in [\ell]} \max_{g_{i_j} \in B_j} g_{i_j} \Big] 
	\end{equation}
	
	The inequality follows from the following observation: For any fixed realization of ${\bf g} \sim N(0,1)^m$, let $i_j$ be the index of the maximum in the $j^{th}$ block. Then there exists a vector in $\widehat{\Sp}_{\ell}$ which is supported on $i_1,\ldots,i_\ell$. Therefore, the dot product would be at least the sum of maximum Gaussians from each of the blocks scaled by $\frac{1}{\sqrt{\ell}}$. The lemma now follows from applying the lower bound from Eq. \ref{fact:gauss_block}.	
	
\end{proof}

	\section{Proof of Lemma \ref{lem:gw_isometricEmbeddings}}		\label{sec:gw_isometry}
	\begin{proof}
		First we prove the upper bound. Let $\Psi:X \mapsto S$ be the $\eps$-isometric embedding map. Given ${\bf g} \sim N(0,\sqrt{1 + \epsilon})^m$ and $\vec{h} \sim N(0,1)^{d}$, we define the gaussian processes $\{G_{\bf x}\}_{{\bf x} \in  X}$ and $\{H_{\bf x}\}_{{\bf x} \in  X}$ as follows 
		
		\begin{eqnarray*}
			G_{\bf x} &\overset{\rm def}{=}& {\bf g}^\top{\bf x} \\
			H_{\bf x} &\overset{\rm def}{=}& {\bf h}^\top\Psi({\bf x})	
		\end{eqnarray*}
		
		Fix ${\bf x},{\bf y} \in X$. Using the $\epsilon$-isometry of $\Psi$ we get: 
		
		\begin{eqnarray*}
			\E_{{\bf h} \sim N(0,1)^{d}} \Big[\Big| H_{\bf x} - H_{\bf y} \Big|^2\Big] 
			&=& \E_{{\bf h} \sim N(0,1)^{d}} \Big[  \Big| {\bf h}^\top\Psi({\bf x}) - {\bf h}^\top\Psi({\bf y}) \Big|^2\Big] \\
			&\overset{1}{=}& \| \Psi({\bf x}) - \Psi({\bf  y}) \|^2 \\
			&\overset{2}{\le}& (1 + \epsilon)\|{\bf x} - {\bf  y} \|^2 \\
			&=& \E_{{\bf g} \sim N(0,{1 + \epsilon})^m} \Big[\Big|G_{\bf x} - G_{\bf y}\Big|^2\Big]\\
		\end{eqnarray*}
		where step $1$ follows from the fact that the variance ${\bf }h$ in the direction of a vector ${\bf v}$ is $\|{\bf v}\|^2_2$, and inequality $2$ follows from the isometric property. Therefore, using Lemma \ref{lem:slepian}, we have 
		
		\begin{equation}
		\E_{{\bf h} \sim N(0,1)^{d}} \Big[\max_{{\bf x} \in S} {\bf g}^\top\Psi({\bf x})  \Big] \le \E_{{\bf g} \sim \sqrt{1 + \epsilon}N(0,1)^m } \Big[\max_{{\bf x} \in S} {\bf g}^\top{\bf x}  \Big] 
		\end{equation}
		which directly gives us $\omega(S) \le \sqrt{1 + \epsilon} \cdot \omega(X) \leq ({1 + {\epsilon}})\omega(X)$. The other direction follows by using the lower bound given by isometry.
		
	\end{proof}

\section{Rotational Invariance of $\phi^{\rm norm}$}

\begin{lemma}			\label{lem:rotational_invariance}
	For any finite set $T \subset \mathcal{S}^{Nn - 1}$, the distribution of $\Phi^{\rm norm}(S)$ is rotation invariant.
\end{lemma}
\begin{proof}
	
	Fix any $N = |T|$ vectors $\{\vec{z}^\prime_1,\ldots,\vec{z}^\prime_N\}$. Let $R:\R^{Nn } \mapsto \R^{Nn}$ be a fixed rotation. With a slight abuse of notation, we shall use $\Pr_\phi(\cdot)$ to denote the pdf of the distribution here. Then,	
	
	\begin{eqnarray*}
		\Pr_\phi\Big(\phi^{\rm norm}(T) = \Big\{\frac{\vec{z}^\prime_1}{\|\vec{z}^\prime_1\|_2},\ldots,\frac{\vec{z}^\prime_N}{\|\vec{z}^\prime_N\|_2}\Big\}\Big) 
		&=& \Pr_\phi\Big(\phi(T) = \{{\vec{z}^\prime_1},\ldots,{\vec{z}^\prime_N}\}\Big) \\
		&\overset{1}{=}& \Pr_\phi\Big(\phi(T) = \{R({\vec{z}^\prime_1}),\ldots,R({\vec{z}^\prime_N})\}\Big) \\ 
		&{=}& \Pr_\phi\Big(\phi^{\rm norm}(T) = \Big\{\frac{R(\vec{z}^\prime_1)}{\|R(\vec{z}^\prime_1)\|_2},\ldots,\frac{R(\vec{z}^\prime_N)}{\|R(\vec{z}^\prime_N)\|_2}\Big\}\Big)  \\
		&\overset{2}{=}& \Pr_\phi\Big(\phi^{\rm norm}(T) = \Big\{\frac{R(\vec{z}^\prime_1)}{\|(\vec{z}^\prime_1)\|_2},\ldots,\frac{R(\vec{z}^\prime_N)}{\|(\vec{z}^\prime_N)\|_2}\Big\}\Big)  \\
		&=& \Pr_\phi\Big(\phi^{\rm norm}(T) = R\Big(\Big\{\frac{\vec{z}^\prime_1}{\|\vec{z}^\prime_1\|_2},\ldots,\frac{\vec{z}^\prime_N}{\|\vec{z}^\prime_N\|_2}\Big\}\Big)\Big)  \\
	\end{eqnarray*}
	
	where step $1$ follows from the rotational invariance of $\phi$, and step $2$ uses the observation that rotating a vector does not change it's $\ell_2$-length. Since the equality holds for any rotation, the statement follows.
\end{proof}

\section{Proof of Lemma \ref{lem:conditional_lb}} 			\label{sec:conditional_lb}

\begin{proof}
	
	The proof uses the more general observation relating conditional expectations to their unconditioned counterparts :
	\begin{proposition}							\label{prop:lowerBound_integ}
		Let $\mathcal{E}$ be an event such that $\Pr(\mathcal{E}) \ge 1 - \eta$. Let $Z:\R^{n} \mapsto \R_{+}$ be a non-negative random variable. Let $ t_0 > 0$ be such that 
		\begin{equation*}
			\int^{\infty}_{t_0} \Pr( Z \ge t) dt \le \alpha
		\end{equation*}
		Then, the following holds true : $\E[ Z | \mathcal{E}] \ge \E[Z] - {\eta t_0} - \alpha$.
	\end{proposition}
	\begin{proof}
		
		We begin by observing that for any event $B$,
		\begin{equation*}
			\Pr( B | \mathcal{E}) \ge \Pr(B) - \eta
		\end{equation*}
		
		which follows from the definition of conditional expectation. Therefore,
		
		\begin{eqnarray*}
			\E[Z|\mathcal{E}] = \int^{\infty}_{0}  \Pr(Z \ge t | \mathcal{E})dt &\ge& \int^{t_0}_{0} \Pr(Z \ge t | \mathcal{E})dt \\
			&\ge& \int^{t_0}_{0} \big(\Pr(Z \ge t ) - \eta \big)dt \\	
			&\ge& \int^{t_0}_{0} \Pr(Z \ge t) dt - \int^{t_0}_{0}\eta dt \\	
			&=& \E[Z] - \alpha - {\eta t_0} 
		\end{eqnarray*}
		
	\end{proof}
	
	We apply the above lemma to our setting: let $\mathcal{E}$ be the event as described in the proof of Lemma \ref{lem:gw_lowerBound} and let $Z$ be the random variable :
	
	\begin{equation*}
		Z := \max_{{\bf x} \in \widehat{\Sp}_{\ell}}{\bf g}^\top\phi({\bf x})
	\end{equation*}
	
	From Eq. \ref{fact:err_bound1}, we know that $\Pr(\neg \mathcal{E}) =  \epsilon_{k,m} = \eta_{k,m}$. Abusing notation, we denote $\omega^* = \omega(\widehat{\Sp}_{\ell})$. 
	Let $t_0 = 4\omega(\widehat{\Sp}_{\ell}) = 4 \omega^*$. For a fixed choice of $\delta > 0$,  let $B_{\delta}$ be the event that $\|{\bf g}\|_2 \le \sqrt{n}(2 + \delta)$. Then, 
	
	\begin{align}
		&\Pr_{{\bf g} \sim N(0,1)^{n},\phi}\Big( \max_{{\bf x} \in \widehat{\Sp}_{\ell}} {\bf g}^\top\phi({\bf x}) \ge (2+\delta)^2\omega^*\Big)  \nonumber\\
		\le&  \Pr(\neg B_\delta) + \Pr_{{\bf g} \sim N(0,1)^{n},\phi}\Big( \max_{{\bf x} \in \widehat{\Sp}_{\ell}} {\bf g}^\top\phi({\bf x}) \ge (2+\delta)^2\omega^*  \Big| B_\delta \Big) \nonumber\\
		\overset{1}{\le}& \exp\Big(-O((1 + \delta)^2n)\Big) + \Pr_{{\bf g} \sim N(0,1)^{n},\phi}\Big( \max_{{\bf x} \in \widehat{\Sp}_{\ell}} {\bf g}^\top\phi({\bf x}) \ge (2+\delta)^2\omega^* | B_\delta \Big) \label{fact:upper_bound1} 
	\end{align}
	
	where inequality $1$ follows by concentration on $\chi^2$ variables. We upper bound the remaining probability term as :
	\begin{align}
		&\Pr_{{\bf g} \sim N(0,1)^{n},\phi}\Big( \max_{{\bf x} \in \widehat{\Sp}_{\ell}} {\bf g}^\top\phi({\bf x}) \ge (2+\delta)^2\omega^* | B_\delta \Big) \nonumber\\
		\overset{2}{\le}& \Pr_{\tilde{\bf g} \sim N(0,1)^{Nn}}\Big( \max_{{\bf x} \in \widehat{\Sp}_{\ell}} \tilde{\bf g}^\top{\bf x} \ge \frac{(2+\delta)^2}{2 + \delta}\omega^* | B_\delta \Big) \nonumber\\ 
		{\le}& \Pr_{\tilde{\bf g} \sim N(0,1)^{Nn}}\Big( \max_{{\bf x} \in \widehat{\Sp}_{\ell}} \tilde{\bf g}^\top{\bf x} \ge (2+\delta)\omega^*  \Big) \\
		\overset{3}{\le}& \exp\Big(-O((1 + \delta)\omega^*)^2\Big) \nonumber\\
		\le& \exp\Big(-O\Big((1 + \delta)^2k\log\frac{m}{k}\Big)\Big)		\label{fact:upper_bound2}
	\end{align}

	Step $2$ can be shown using arguments identical to the ones used in steps \ref{pf:step1}-\ref{pf:step2} in the proof of Lemma \ref{thm:sparse_approx}, and step $3$ follows from gaussian concentration. {Now we pr}oceed to upper bound the quantity $\alpha$ (as in Proposition \ref{prop:lowerBound_integ} ):
	
	\begin{eqnarray}
	\int_{t_0}^{\infty} \Pr(Z \ge t) dt &=& \int_{0}^{\infty}  \Pr_{{\bf g} \sim N(0,1)^{n},\phi}\Big( \max_{{\bf x} \in \widehat{\Sp}_{\ell}} {\bf g}^\top\phi({\bf x}) \ge 4\omega^* + t\Big) dt 		\nonumber\\
	&\overset{1}{=}&\int_{0}^{\infty} 2(2 + \delta)\omega^* \Pr_{{\bf g} \sim N(0,1)^{n},\phi}\Big( \max_{{\bf x} \in \widehat{\Sp}_{\ell}} {\bf g}^\top\phi({\bf x}) \ge (2+\delta)^2\omega^* \Big) d\delta 		\nonumber\\
	&\overset{2}{\le}& \int_{0}^{\infty} (2 + \delta)\omega^* \exp\Big(-O((1 + \delta)^2k\log\frac{m}{k})\Big) d\delta  \nonumber\\
	&\le& \int_{0}^{\infty} \exp\Big(-O((1 + \delta)^2k\log\frac{m}{k})\Big) d\delta  = \alpha_{k,m}         \label{fact:upper_bound3}
	\end{eqnarray}
	
	where step $1$ is a change of variables argument where we set $t = (2 + \delta)^2\omega^* - 4\omega^*$, and the second step follows from by combining upper bounds from \ref{fact:upper_bound1} and \ref{fact:upper_bound2}.
	
	Plugging in the upper bounds from Equations \ref{fact:upper_bound1},\ref{fact:upper_bound2} and \ref{fact:upper_bound3}, we get
	
	\begin{equation*}
		\E[Z | \mathcal{E}] \ge \E[Z] - \alpha_{k,m} - 4(\omega^*)\eta_{k,m} =  \E[Z] - \gamma_{k,m}
	\end{equation*}
	
	where $\gamma_{k,m}$ decays exponentially in $k$.
	
\end{proof}

\section{Proof of Lemma \ref{lem:l_infty}}		\label{sec:l_infty}

	The proof uses the observation that for ${\bf R} \sim \mathbbm{O}_d$, for any $i \in [d]$ marginal distribution of the vector ${\bf r}_i$ is that of a uniformly random vector drawn from $\mathcal{S}^{d-1}$ (c.f., Exercise $5$ \cite{Vershynin11}). Therefore, it suffices to show large probability upper bounds for a single random vector ${\bf r} \sim \mathcal{S}^{d-1}$, which can then be used to complete the proof by a union bound argument.

	\vspace{0.5cm}.
	
	\noindent{\it Concentration for random unit vectors}: 	Let $C > 0$ be a constant which is fixed later. The first step follows from replacing the unit vector by a normalized gaussian vector:
	\begin{eqnarray*}
		\Pr_{{\bf r} \sim \mathcal{S}^{d-1}}\Bigg[ \max_{{\bf x} \in S} {\bf r}^\top{\bf x} \ge C\omega(S)/\sqrt{d} \Bigg] &=& \Pr_{{\bf g} \sim N(0,1)^{d}}\Bigg[ \max_{{\bf x} \in S} {\bf g}^\top{\bf x} \ge  C\omega(S)\frac{\|{\bf g}\|}{\sqrt{d}}  \Bigg] \\
		&\le& \Pr_{{\bf g} \sim N(0,1)^{d}}\Bigg[ {\bf g}^\top{\bf x} \ge  C\omega(S) \frac{\|{\bf g}\|}{\sqrt{d}} \Big| \|{\bf g}\| \ge \sqrt{d}/2\Bigg] + \Pr_{{\bf g} \sim N(0,1)^{d}}\Bigg[ \|{\bf g}\| \le \sqrt{d}/2\Bigg] \\
		&\le& 2\Pr_{{\bf g} \sim N(0,1)^{d}}\Bigg[{\bf g}^\top{\bf x} \ge  C\omega(S)/2 \Bigg] + \Pr_{{\bf g} \sim N(0,1)^{d}}\Bigg[ \|{\bf g}\| \le \sqrt{d}/2\Bigg] \\
		&\le& 4\max\bigg(\exp(-C^\prime\omega^2(S)), \exp(-C^\prime d)\bigg) 
	\end{eqnarray*}
	
	\noindent where the first term is upper bounded using Lemma \ref{lem:gconc}, and the second term is bounded using $\chi^2$-concentration.\\
	
	\noindent{\it Concentration for random rotations}: We now extend the above concentration bound to an expectation bound for random rotations. Let $\mathcal{E}$ denote the event for ${\bf R} \sim \mathbbm{O}_d$, there exists ${\bf x} \in S$ such that $\|{\bf R}{\bf x}\|_\infty \ge C\omega(S)/\sqrt{d}$.
	
	\begin{eqnarray*}
	\Pr_{{\bf R} \sim \mathbbm{O}_d} \Bigg[ \max_{{\bf x} \in S} \|{\bf R}{\bf x}\|_\infty > C\omega(S)/\sqrt{d}\Bigg] 
	&=& \Pr_{{\bf R} \sim \mathbbm{O}_d} \Bigg[  \max_{i \in [d]}\max_{{\bf x} \in S} {\bf r}^\top_i{\bf x} > C\omega(S)/\sqrt{d}\Bigg] \\
	&\le&  \sum_{i \in [d]}\Pr_{{\bf R} \sim \mathbbm{O}_d} \Bigg[  \max_{{\bf x} \in S} {\bf r}^\top_i{\bf x} > C\omega(S)/\sqrt{d}\Bigg] \\
	&\le& 4\max\bigg(\exp(-C^\prime\omega^2(S)), \exp(-C^\prime d)\bigg) \le \frac{1}{2}
	\end{eqnarray*}
	
	where the last step follows from the fact that $d >> \log d$ and by choice $C^\prime\omega(S) \ge 10\log d$, when $C$ is chosen to be large enough.

\section{Analysis for the Dimensionality Tester}

We state the definition of $\epsilon$-approximate rank of a matrix, as defined in \cite{ALSV13}:

\begin{definition}[Approximate Rank]			\label{defn:appx_rank}
Given a matrix ${\bf Y} \in \R^{m \times n}$ and an $\epsilon > 0$, the $\epsilon$-approximate rank of the matrix \emph{(denoted by ${\rm rank}_\epsilon({\bf Y})$)} is defined as follows:
\begin{equation}
	{\rm rank}_{\epsilon}({\bf Y}) = \min\bigg(\Big\{{\rm rank}(\hat{\bf Y}): \hat{\bf Y} \in \R^{m \times n}, \|{\bf Y} - \hat{\bf Y}\|_{\infty,\infty} \le \epsilon\Big\}\bigg)
\end{equation}

where $\|\cdot\|_{\infty,\infty}$ is norm defined as the largest absolute value of an entry in the matrix.
\end{definition}

We first prove Lemma \ref{lem:appx-rank-bound} which relates the approximate rank of a matrix in terms of the gaussian width, and use that to analyze the tester.

\begin{lemma}			\label{lem:appx-rank-bound}
	For a matrix ${\bf Y} \in \mathbbm{R}^{d \times n}$, where $\|{\bf y}_i\| = 1 \enskip \forall \enskip i \in [n]$,  the following holds:
	\begin{equation}
	{\rm rank}_{\epsilon}({\bf Y}) \le O\bigg(\frac{1}{\epsilon^2}\max\Big(\omega^2({\bf Y}),\log d\Big)\bigg)
	\end{equation}	
	for any $\epsilon \ge O(1/\sqrt{d})$.
\end{lemma}
\begin{proof}
 Let $Y = \{{\bf y}_1,\ldots,{\bf y}_n\}$ be the set of columns from the matrix ${\bf Y}$. Let $Y_0 = Y \cup {\rm I}_d$ where ${\rm I}_d$ is the set of standard basis vectors $\{{\bf e}_i\}_{i \in [d]}$. It is known that gaussian width is subadditive, and therefore
\begin{equation}
\omega(Y_0) \le \omega(Y) + \omega({\rm I}_d) \le 2\max\Big(\omega(Y),2\sqrt{\log d}\Big)
\end{equation}

Let $d^\prime = \frac{16C}{\epsilon^2}\max\Big(\log d, \omega^2(Y)\Big)$ where $C$ is the constant given by the generalized JL-lemma and let ${\bf G} \sim \frac{1}{\sqrt{d^\prime}}N(0,1)^{d^\prime \times d}$. Then with high probability, ${\bf G}(Y_0)$ is $\epsilon$-isometric to $Y_0$. For every  $i \in [d]$ and $j \in [n]$, we observe that: 

\begin{itemize}
	\item[1.] $1 - \epsilon \ge \|{\bf G}{\bf e}_i\|^2,\|{\bf G}{\bf y}_j\|^2 \le 1 + \epsilon$
	\item[2.] $(1 - \epsilon)\|{\bf e}_i - {\bf y}_j\|^2 \le \|{\bf G}{\bf e}_i - {\bf G}{\bf y}_j\|^2 \le (1 + \epsilon)\|{\bf e}_i - {\bf y}_j\|^2$ which in turn implies that $|\langle {\bf G}{\bf e}_i ,{\bf G}{\bf y}_j \rangle - \langle {\bf e}_i,{\bf y}_j \rangle| \le O(\epsilon) $.      
\end{itemize}

Let ${\bf Y}^\prime = {\bf G}^\top{\bf G}{\bf Y}$. Since the above observation is true for any $i \in [d],j \in [n]$, it follows that ${\bf Y}^\prime$ is entry wise $O(\epsilon)$-close to ${\bf Y}$, and by construction ${\rm rank}({\bf Y}^\prime) \le d^\prime$. Hence, the claim follows.

\end{proof}

Using the above lemma, we now show completeness and soundness for the tester:

\begin{proof}[Proof of \cref{thm:test_dim}]
	Let $S$ denote the set $\{\vec{y}_1, \dots, \vec{y}_p\}$.
	The tester obtains $\hat{\omega}$ that approximates $\omega(S)$ to an additive error of $\sqrt{k}$ and accepts iff $\hat{\omega} \leq 2 \sqrt{k}$. By \cref{lem:estim},  the tester requires  $O(p \log \delta^{-1})$ linear queries to obtain $\hat{\omega}$.
	
	If $\dim(S) \leq k$, then by \cref{lem:width}, $\hat{\omega} \leq 2\sqrt{k}$ with probability at least $1-\delta$, so that the tester accepts with the same probability.
	
	If the tester accepts, then with probability at least $1-\delta$, $\omega(S) \leq 3\sqrt{k}$. Therefore, from Lemma \ref{lem:appx-rank-bound}, we have ${\rm rank}_\epsilon(Y) \le O(k/\epsilon^2)$ which completes the proof. 
\end{proof}

\section{On the relationship between RIP and Incoherence}\label{appendix:incoherence-rip}

Even though our results are stated in terms of dictionaries which satisfy RIP, they can be stated in terms of incoherence as well. This is because the incoherence{\footnote{Here incoherence is stated in dimension free terms i.e., $|\langle {\bf a}_i , {\bf a}_j \rangle | \le \mu $ for every $i \ne j$}} and RIP constants of the dictionary matrix are roughly equivalent. We formalize this observation in the following lemma:

\begin{proposition}\label{prop:ripinco}
	Let ${\bf A} \in \mathbbm{R}^{d \times m}$ be a matrix with $\|{\bf a}\|_i = 1$ for every $i \in [m]$. Then, 
	\begin{itemize}
		\item If ${\bf A}$ is $(2k,\zeta)$-RIP then it is $\zeta$-incoherent.
		\item If ${\bf A}$ is $\mu$-incoherent, then it is $(2k,4k\mu)$-RIP
	\end{itemize}
\end{proposition}
\begin{proof}
	Suppose ${\bf A}$ is $(2k,\zeta)$-RIP. Then for any $i,j \in [m]$ 
	
	\begin{equation}
	|\langle {\bf a}_i, {\bf a}_j \rangle | = 1 - \frac{\|{\bf a}_i - {\bf a}_j\|^2}{2} \overset{1}{\in} 1 - (1 \pm \zeta) = \pm \zeta
	\end{equation}	
	
	\noindent where $1$ follows using the RIP guarantee. On the other hand, let ${\bf A}$ be $\mu$-incoherent. Then for any $S \subset [m]$ of size $2k$ let ${\bf M} = {\bf A}_S^\top{\bf A}_S$ where ${\bf A}_S$ is the submatrix induced by columns in $S$. Then we observe that $M_{ii} = \|{\bf a}_i\|^2 = 1$ for every $i \in [2k]$ and off-diagonal entries satisfy $|M_{ij}| \le \mu $. Therefore, using the Gershgorin's disk theorem $\lambda({\bf M}) \in [1 \pm 2\mu k]$. Therefore for every ${\bf x}$ supported on $S$ we have $\|{\bf A}{\bf x}\|^2 \in ( 1 \pm 2\mu k)^2\|{\bf x}\|^2 \in ( 1 \pm 4\mu k)\|{\bf x}\|^2$. Since this is true for any arbitrary $2k$-sized subset $S$, the result follows.
	
\end{proof}

Note that quantitatively, incoherence is a stronger property than RIP since $\mu$-incoherence implies $(2k,4k\mu)$-RIP but $(2k,4k\mu)$-RIP only implies $4k\mu$-incoherence. Naturally, Theorem \ref{thm:test_rip} can be restated in terms of incoherent linear transformations as well.

\end{document}